\documentclass{lmcs}
\pdfoutput=1

\usepackage{lastpage}
\renewcommand{\dOi}{14(4:1)2018}
\lmcsheading{}{1--\pageref{LastPage}}{}{}%
{Jan.~09,~2018}{Oct.~16,~2018}{}

\usepackage[T1]{fontenc}
\usepackage[latin9]{inputenc}
\usepackage{color}
\usepackage{amsmath}
\usepackage{amssymb}
\usepackage{nicefrac}
\usepackage{hyperref}

\newcommand{\bi}{\begin{itemize}}
 \newcommand{\ei}{\end{itemize}}
 \newcommand{\be}{\begin{enumerate}}
 \newcommand{\ee}{\end{enumerate}}
 \newcommand{\bd}{\begin{description}}
 \newcommand{\ed}{\end{description}}
 \newcommand{\bs}{\bigskip}
 \newcommand{\ms}{\medskip}
  \newcommand{\w}{\wedge}

\makeatletter

\usepackage{footnote}

\global\long\def\setin#1#2{{\left\{  \right.\negthickspace\negthickspace\negthickspace\circ}#1\mid#2{\circ\negthickspace\negthickspace\negthickspace\left.\right\}  }}
\global\long\def\setfin#1{{\left\{  \right.\negthickspace\negthickspace\negthickspace\circ}#1{\circ\negthickspace\negthickspace\negthickspace\left.\right\}  }}
\global\long\def\class#1#2{{\left\{  \right.\negthickspace\negthickspace\negthickspace\circ}#1\hat{\,\mid\,}#2{\circ\negthickspace\negthickspace\negthickspace\left.\right\}  }}

 \global\long\def\setout#1#2{\left\{  #1\mid#2\right\}  }
 \global\long\def\imp{\rightarrow}
\global\long\def\saferst{\succ}

 \global\long\def\safe{\succ}
 
 \global\long\def\i{\succ}

\global\long\def\lng{\mathcal{L}_{RST}^{\{HF\}}}
 
 \global\long\def\lngc{\mathcal{L}_{RST}^{C}}
 
\global\long\def\inw#1{\left\Vert #1\right\Vert _{v}}
\global\long\def\inwa#1{\left\Vert #1\right\Vert _{v\left[x:=a\right]}}
\global\long\def\inwv#1{\left\Vert #1\right\Vert _{v}^{W}}

\usepackage{enumitem}

\makeatother

\begin{document}

\title{Applicable Mathematics in a Minimal Computational Theory of Sets}

\author[A.~Avron]{Arnon Avron}  
\address{Tel Aviv University,  Tel-Aviv, Israel}        
\email{aa@post.tau.ac.il}  

\author[L.~Cohen]{Liron Cohen}  
\address{Cornell University, Ithaca, NY, USA}   
\email{lironcohen@cornell.edu}  
\begin{abstract}
In previous papers on this project a general static logical framework for formalizing
and mechanizing set theories of different strength was suggested, and the power of 
some predicatively acceptable theories in that framework was explored. 
In this work we first improve that framework 
by enriching it with means for coherently extending by definitions
its theories, without destroying its static nature 
or violating any of the principles on which it is based. Then we turn to
investigate within the enriched framework the power of 
the {\em minimal} (predicatively acceptable) theory in it
that proves the existence of infinite sets. We show that that theory 
is a computational theory, in the sense that
every element of its minimal transitive model 
is denoted by some of its closed terms. (That model happens to be the second universe 
in Jensen's hierarchy.)  
Then we show that  already this minimal theory
suffices for developing very large portions (if not all) of scientifically  applicable mathematics.
This requires treating the collection of real numbers as a proper class,
that is: a unary predicate which can be introduced in the theory
by the static extension method described in the first part of the paper.
\end{abstract}

\maketitle

\section{Introduction}

Formalized mathematics and mathematical knowledge management (MKM)
are extremely fruitful and quickly expanding fields of research at
the intersection of mathematics and computer science
(see, e.g.,
 \cite{avigad2014formally,campbell2008intelligent,kamareddine2003thirty}).
The declared goal of these fields is to develop computerized systems
that effectively represent all important mathematical knowledge and
techniques, while conforming to the highest standards of mathematical
rigor. At present there is no general agreement
what should be the best framework for this task.
However, since most mathematicians view {\em set theory}
as the basic foundation of mathematics, formalized set theories should
certainly be taken as one of the most natural choice.\footnote{Already in 
\cite{cantone2001set} it was argued that
``a main asset gained from Set theory is the ability to base reasoning
on just a handful of axiom schemes which, in addition to being conceptually
simple (even though surprisingly expressive), lend themselves to good
automated support''. More recently, H. Friedman wrote
(in a message on FOM on  Sep 14, 2015): 
``I envision a large system
and various important weaker subsystems. Since so much math can be
done in systems much weaker than ZFC, this should be reflected in
the choice of Gold Standards. There should be a few major Gold Standards
ranging from Finite Set Theory to full blown ZFC''.}
\footnote{Notable set-based
automated provers are Mizar \cite{rudnicki1992overview},
Metamath \cite{Megill97metamath}, and Referee (aka AetnaNova)~\cite{Omodeo2006,SCO11}.}

In \cite{avron2008framework,avron2010new} a logical framework for
developing and mechanizing set theories was introduced. Its key properties
are that it is based on the usual (type-free) set theoretic language
and makes extensive use of {\em abstract set terms.} Such terms are extensively used
of course in all modern texts in all areas of mathematics
(including set theory itself). Therefore their availability is indispensable 
for the purpose of mechanizing
real mathematical practice and for automated or
interactive theorem proving in set theories.
Accordingly, most of the computerized systems for set theories indeed allow {\em dynamic} ways of
introducing abstract set terms. The great advantage of the framework
of \cite{avron2008framework,avron2010new} is that
it does so in a {\em static} way, so the task
of verifying that a given term or  formula in it is well-formed
is decidable, easily mechanizable, and completely separated from any task
connected with proving theorems (like finding proofs or checking
validity of given ones).
Furthermore, this framework enables the use of different logics and set theories
of different strength.  This modularity of the system has been
exploited in \cite{Cohen2014QED}, where a hierarchy of set theories
for formalizing different levels of mathematics within this framework
was presented. 

The current paper is mainly devoted to one  very basic theory, $RST_{HF}^m$, 
from the above-mentioned hierarchy, and to its 
minimal model. The latter is shown to be the universe $J_2$
in Jensen's hierarchy \cite{Jensen1972fine}. Both $RST_{HF}^m$ and $J_2$
are \emph{computational} (in a precise sense defined below). 
With the help of the formal framework
of \cite{avron2008framework,avron2010new,Cohen2014QED} they can therefore
be used 
to make explicit the potential computational content of set theories
(first suggested and partially demonstrated in \cite{cantone2001set}).
Here we show that they also suffice for
developing large portions of scientifically applicable mathematics \cite{feferman1992little},
especially analysis.\footnote{The thesis that $J_{2}$ is sufficient for core mathematics was first
put forward in \cite{weaver2005analysis}.} In 
\cite{feferman1964systems,feferman1968systems,feferman1992little} it was forcefully 
argued by Feferman that scientifically applicable mathematics, the mathematics that is 
actually indispensable to present-day natural science, can be developed using only 
predicatively acceptable mathematics. We provide here further support to this claim, 
using a much simpler framework and by far weaker theory than those employed by Feferman.

The restriction to a minimal framework has
of course its price. Not all of the standard mathematical structures 
can be treated as elements of $J_{2}$. (The real line is a case in point.) Hence we
have to handle such objects in a different manner. To do this,
 we first enrich the  framework used in 
\cite{avron2008framework,avron2010new,Cohen2014QED} 
with means for coherently extending by definitions
theories in it, without destroying its static nature, 
or violating any of the principles on which it is based. (This step is 
a very important improvement of the framework  on each own right.)
This makes it possible to introduce the collection of real numbers 
in $RST_{HF}^m$ as a {\em proper class},
that is: a legal defined unary predicate to which no closed term 
of $RST_{HF}^m$ corresponds.
(Classes are introduced here into the formal framework
of \cite{avron2008framework,avron2010new,Cohen2014QED}
for the first time.)

\bs

The paper is organized as follows: In 
Section \ref{sec:Preliminaries:-The-First-Order}
we review the formal framework and the way various standard set theoretical
notions have been introduced in in it. We also define in this
section the notions of computational theory and 
universe, and describe the 
computational theories which are minimal within the framework
(as well as the corresponding minimal universes). 
Section \ref{sec:Static-Extensions-by}
is dedicated to the introduction of standard extensions by definitions
of the framework, done in a static way. The notion of a class is then introduced
as a particular case, and is used for handling global relations and functions
in the system. In Section~\ref{sec:The-Natural-Numbers} we introduce the natural
numbers in the system. Unlike in \cite{Cohen2014QED}, this is done here using
an absolute characterization of the property of being a natural number,
and without any appeal to $\in$-induction.
In Section \ref{sec:Real-Analysis} we turn to 
real analysis, and demonstrate how it can be developed in our minimal
computational framework, although the reals are a proper class in
it. This includes the introduction of the real line and  real
functions, as well as formulating and proving classical results
concerning these notions.\footnote{\label{fn:counterpart}A few of the claims in 
Section \ref{sec:Real-Analysis}
have counterparts in \cite{Cohen2014QED}. However, the models used in that paper 
are based on universes which are more extensive than the minimal one
which is studied here. Hence the development and proofs there were much simpler. 
In particular: there was no need in \cite{Cohen2014QED} to use proper classes,
as is essential here. Another, less crucial but still important, difference is that 
unlike in \cite{Cohen2014QED}, the use of $\in$-induction is completely
avoided at the present paper.}
Section \ref{sec:Further-Research} concludes with directions for future 
continuation of the work.

\section{\label{sec:Preliminaries:-The-First-Order}The Formal System and its Minimal Model}
\subsection{\label{sec:The-First-Order}Preliminaries: the Framework and the Main Formal System}

\begin{nota}
To avoid confusion, the parentheses $\setfin{\,}$
are used in our formal languages\emph{,} for constructing
abstract set terms in it, while in the meta-language
we use the ordinary  $\left\{ \,\right\} $.\footnote{To be extremely precise, 
we should have also used different notations in the formal languages
and in the meta-language for $\in$ and $=$, as well
as for many other standard symbols which are used below.
However, for readability we shall not do so, and trust the
reader to deduce the correct use from the context.}
We use the letters $X,Y,Z,...$ for
collections; $\Phi,\Theta$ for finite sets of variables; and $x,y,z,...$
for variables in the formal language. $Fv(exp)$ denotes the set of
free variables of $exp$, and  
$\ensuremath{\varphi\left[\nicefrac{t_{1}}{x_{1}},\ldots,\nicefrac{t_{n}}{x_{n}}\right]}$
denotes the result of simultaneously substituting $t_{i}$ for $x_{i}$
in $\varphi$. When the identity of $t$ and $x$ is clear from the context,
we just write $\varphi(t)$ instead of 
$\ensuremath{\varphi\left[\nicefrac{t}{x}\right]}$.
\end{nota}

One of the foundational questions in set theory is which formulas
should be excluded from defining sets by an abstract term of the form
$\left\{ x\mid\varphi\right\} $ in order to avoid the paradoxes of
naive set theory. 
Various set theories provide different answers to this
question,  which are usually based on \emph{semantical} 
considerations (such as the limitation of size doctrine 
\cite{fraenkel1973foundations,hallett1984cantorian}).
Such an approach is not very  useful for the purpose of mechanization. In this
work we use instead the general \emph{syntactic} methodology of safety
relations developed in \cite{avron2008framework,avron2010new}.
A safety relation is a syntactic relation between formulas and sets
of variables. The addition of a safety relation to a logical system
allows to use in it statically defined abstract
set term of the form $\left\{ x\mid\varphi\right\} $, provided that
$\varphi$ is safe with respect to $\{x\}$. Intuitively, a statement
of the form ``$\varphi$ is \emph{safe} with respect to $\left\{ y_{1},...,y_{k}\right\} $\textquotedblright{},
where $Fv(\varphi)=\left\{ x_{1},...,x_{n},y_{1},...,y_{k}\right\} $,
has the meaning that for every ``accepted\textquotedblright{} sets
$a_{1},...,a_{n}$, the collection $\left\{ \left\langle y_{1},...,y_{k}\right\rangle \mid\varphi(a_{1},...,a_{n},y_{1},...,y_{k})\right\} $
is an ``accepted\textquotedblright{} set, which is constructed from
the previously ``accepted\textquotedblright{} sets
$a_{1},...,a_{n}$ (see discussion below for further details).

\begin{defi}
\label{def: safety}Let $C$ be a finite set of constants. The language
\emph{$\lngc$} and the associated safety relation $\saferst$ are
simultaneously defined as follows: 
\begin{itemize}
\item Terms:
\begin{itemize}
\item Every variable is a term.
\item Every $c\in C$ is a term (taken to be a constant). 
\item If $x$ is a variable and $\varphi$ is a formula such that $\varphi\saferst\left\{ x\right\} $,
then $\setin x{\varphi}$ is a term ($Fv\left(\setin x{\varphi}\right)=Fv\left(\varphi\right)-\left\{ x\right\} $).
\end{itemize}
\item Formulas:
\begin{itemize}
\item If $s,t$ are terms, then $t=s$, $t\in s$ are atomic formulas. 
\item If $\varphi,\psi$ are formulas and $x$ is a variable, then $\neg\varphi,\left(\varphi\wedge\psi\right),\left(\varphi\vee\psi\right)$,
$\exists x\varphi$ are formulas.\footnote{\label{foot:forall}Our official language does not include 
$\forall$ and $\rightarrow$. However, 
since the theory studied in this paper  is based on classical logic,
 we take here  $\forall x_{1}...\forall x_{n}\left(\varphi\rightarrow\psi\right)$
as an abbreviation for $\neg\exists x_{1}...\exists x_{n}\left(\varphi\wedge\neg\psi\right)$.} 
\end{itemize}
\item The safety relation $\saferst$:
\begin{itemize}
\item If $\varphi$ is an atomic formula, then $\varphi\saferst\emptyset$. 
\item If $t$ is a term such that $x\notin Fv\left(t\right)$, and 
$\varphi\in\left\{ x\neq x,x\in t,x=t,t=x\right\} $,
then $\varphi\saferst\left\{ x\right\} $.
\item If $\varphi\saferst\emptyset$, then $\neg\varphi\saferst\emptyset$. 
\item If $\varphi\saferst\Theta$ and $\psi\saferst\Theta$, then $\varphi\vee\psi\saferst\Theta$. 
\item If $\varphi\saferst\Theta$, $\psi\saferst\Phi$ and $\Phi\cap Fv\left(\varphi\right)=\emptyset$
or $\Theta\cap Fv\left(\psi\right)=\emptyset$, then $\varphi\wedge\psi\saferst\Theta\cup\Phi$. 
\item If $\varphi\saferst\Theta$ and $y\in\Theta$, then $\exists y\varphi\saferst\Theta-\left\{ y\right\} $.
\end{itemize}
\end{itemize}
\end{defi}
\noindent \emph{Notation}. We take the usual definition of $\subseteq$
in terms of $\in$, according to which $t\subseteq s\safe\emptyset$.
\begin{defi}
\label{def:system} An $RST$-theory \footnote{`RST' stands for Rudimentary
Set Theory. See Theorem \ref{thm:Arnon-1-1}  below.}
is a classical
first-order system with variable binding term operator (\cite{corcoran1972variable}),  
in a language of the form \emph{$\lngc$}, which includes
the following axioms:
\begin{itemize} 
\item Extensionality:~~~ $\forall z\left(z\in x\leftrightarrow z\in y\right)\imp x=y$ 
\item Comprehension Schema:~~~ $\forall x\left(x\in\setin x{\varphi}\leftrightarrow\varphi\right)$ 
\end{itemize} 
\end{defi}

\begin{lem}
\label{lem:example}\cite{avron2010new} 
The following notations are
available  (i.e. they can be introduced as abbreviations
and their basic properties are provable) in every $RST$-theory:
\begin{itemize}
\item $\emptyset:=\setin x{x\neq x}$.
\item $\setfin{t_{1},...,t_{n}}:=\setin x{x=t_{1}\vee...\vee x=t_{n}}$,
where $x$ is fresh. 
\item $\left\langle s,t\right\rangle :=\setfin{\setfin s,\setfin{s,t}}$.
$\left\langle t_{1},...,t_{n}\right\rangle :=\left\langle \left\langle t_{1},...,t_{n-1}\right\rangle ,t_{n}\right\rangle $.
\item $\pi_{1}\left(t\right):=\setin x{\exists y.t=\left\langle x,y\right\rangle },\,\pi_{2}\left(t\right):=\setin y{\exists x.t=\left\langle x,y\right\rangle }$.
\item $\setin{x\in t}{\varphi}:=\setin x{x\in t\wedge\varphi}$, provided
$\varphi\saferst\emptyset$ and $x\notin Fv\left(t\right)$. 
\item $\setin t{x\in s}:=\setin y{\exists x.x\in s\wedge y=t}$, where $y$
is fresh and $x\notin Fv\left(s\right)$. 
\item $s\times t:=\setin x{\exists a\exists b.a\in s\wedge b\in t\wedge x=\left\langle a,b\right\rangle }$,
where $x,a,b$ are fresh.
\item $s\cup t:=\setin x{x\in s\vee x\in t}$, where $x$ is fresh.
\item $s\cap t:=\setin x{x\in s\wedge x\in t}$, where $x$ is fresh. 
\item $\cup t:=\setin x{\exists y\in t.x\in y}$, where $x,y$ are fresh.
\item $\cap t:=\setin x{x\in\cup t\wedge\forall y\in t.x\in y}$, where
$x,y$ are fresh.
\item $\iota x.\varphi:=\bigcup\setin x{\varphi}$, provided $\varphi\saferst\left\{ x\right\} $.\footnote{Due to the Extensionality Axiom, if $\varphi\saferst\left\{ x\right\} $,
then the term above for $\iota x.\varphi$ denotes $\emptyset$ if
there is no set which satisfies $\varphi$, and it denotes the union
of all the sets which satisfy $\varphi$ otherwise. In particular:
this term has the property that if there is exactly one set which
satisfies $\varphi$, then $\iota x.\varphi$ denotes this unique
set since $\cup\left\{ a\right\} =a$. Note that the definition of
$\iota x.\varphi$ taken here is simpler than the definition used
in \cite{avron2010new}, which was $\cap\setin x{\varphi}$ (where
some caution was taken so that the term is always well defined). }
\item $\lambda x\in s.t:=\setin y{\exists x.x\in s\wedge  y=\left\langle x,t\right\rangle }$, provided $x\notin Fv\left(s\right)$.
\item $Dom\left(t\right):=\setin x{\exists z\exists v\exists y.z\in t\wedge v\in z\wedge y\in v\wedge x\in v\wedge z=\left\langle x,y\right\rangle }$,
($z,v,x, y$ fresh). 
\item $Im\left(t\right):=\setin y{\exists z\exists v\exists x.z\in t\wedge v\in z\wedge y\in v\wedge x\in v\wedge z=\left\langle x,y\right\rangle }$,
($z,v,x, y$ fresh).
\end{itemize}
\end{lem}

\begin{lem}
\label{prop:pairs}\cite{avron2008framework}
There are formulas, $t\check{=}\left\langle r,s\right\rangle $
and $\left\langle r,s\right\rangle \check{\in}t$ in $\lng$ such that:
\begin{enumerate}
\item $t\check{=}\left\langle x,s\right\rangle \saferst\left\{ x\right\} $,
$t\check{=}\left\langle s,x\right\rangle \saferst\left\{ x\right\} $
and $t\check{=}\left\langle x,y\right\rangle \saferst\left\{ x,y\right\} $
for $x,y\notin Fv\left(t\right)$.
\item $\left\langle x,s\right\rangle \check{\in}t\saferst\left\{ x\right\} $,
$\left\langle s,x\right\rangle \check{\in}t\saferst\left\{ x\right\} $
and $\left\langle x,y\right\rangle \check{\in}t\saferst\left\{ x,y\right\} $
for $x,y\notin Fv\left(t\right)$.
\item $r=\left\langle s,t\right\rangle \leftrightarrow r\check{=}\left\langle s,t\right\rangle $
is provable in every $RST$-theory.
\end{enumerate}
\end{lem}

\begin{defi}
\label{def: RST-HF}
\ 
\be
\item $RST^m$ is the minimal $RST$-theory. In other words: $RST^m$ is the theory
in $\mathcal{L}_{RST}^{\emptyset}$ whose axioms are those given in 
Definition~\ref{def:system}.
\footnote{ $RST^m$ can be shown to be equivalent to 
Gandy\textquoteright s basic set theory \cite{gandy1974set}.}
\item $RST_{HF}^m$ is the $RST$-theory in $\mathcal{L}_{RST}^{\{HF\}}$ in which 
the following axioms are added to those given in
Definition~\ref{def:system}:

\begin{itemize}
\item $\emptyset\in HF$ 
\item $\forall x\forall y\left(x\in HF\wedge y\in HF\rightarrow x\cup\setfin y\in HF\right)$
\item $\forall y\left(\emptyset\in y\wedge\forall v,w\in y.v\cup\setfin w\in y\rightarrow HF\subseteq y\right)$
\end{itemize}
\ee
\end{defi}


\noindent
{\em Discussion.}
\begin{itemize}
\item
In \cite{avron2008framework} it was suggested that the computationally
meaningful instances of the Comprehension Axiom are those which determine
the collections they define in an absolute way, independently of any
\textquotedblleft surrounding universe\textquotedblright . 
In the context of set theory, a formula $\varphi$ is ``computable''
w.r.t. $x$ if the collection $\setout x{\varphi\left(x,y_{1},...,y_{n}\right)}$
is completely and uniquely determined by the identity of the parameters
$y_{1},...,y_{n}$, and the identity of other objects referred to
in the formula (all of which are well-determined beforehand). Note
that $\varphi$ is computable for $\emptyset$ iff it is absolute
in the usual sense of set theory. In order to translate this idea
into an exact, \emph{syntactic} definition, the safety relation is
used. Thus, in an $RST$-theory
only those formulas which are safe with respect to $\left\{ x\right\} $
are allowed in the Comprehension Scheme.
It is not difficult to see  that the safety relation $\saferst$ used in
an $RST$-theory indeed possesses
the above property.\footnote{Recently it was  shown \cite{ALL18}  that up to logical equivalence,
and as long as we restrict ourselves to the basic first-order language, 
the converse holds as well. 
It is not known yet whether this is true also in the presence of abstract set terms.}
Thus the formula $x\in y$ should be safe
w.r.t. $\left\{ x\right\} $ (but not w.r.t. $\left\{ y\right\} $),
since if the identity of $y$ is computationally acceptable as a set,
then any of its elements must be previously accepted as a set, and
$\left\{ x\,|\,x\in y\right\} =y$. Another example is given by
the clause for negation. The intuitive meaning of $\left\{ x\,|\,\neg\varphi\right\} $
is the complement (with respect to some universe) of $\left\{ x\,|\,\varphi\right\} $,
which is not in general computationally accepted. However, if $\varphi$
is absolute, then so is its negation. 

\item
$RST^m$ and $RST_{HF}^m$ differ from the systems $RST$ and $RST_{HF}$
used in \cite{Cohen2014QED} with respect to the use of $\in$-induction. 
In principle, $\in$-induction does not seem to
be in any conflict with the notion of a computational theory, since
 it only imposes further restrictions on the collection of
acceptable sets.  Accordingly, it was indeed 
adopted and used in \cite{avron2008framework,avron2010new,Cohen2014QED}.
 Nevertheless, in order not to impose unnecessary constraints
on our general framework, and in particular to allow to develop in it
set theories which adopt  the anti-foundation
axiom AFA, $\in$-induction is not included in $RST^m$ and $RST_{HF}^m$.

\item
It is not difficult to prove that 
$\mathcal{HF}$, the set of all hereditary finite sets,
is a model of $RST^m$. (In fact, it is the minimal one.)
It follows that the set $\mathbb{N}$ of the natural numbers
is not definable as a set in $RST^m$. To solve this problem,
the special constant $HF$ was added in $RST_{HF}^m$, together with appropriate
axioms. (These axioms  replace in $RST_{HF}^m$ the usual infinity axiom
of $ZF$.) The intended interpretation of the new constant $HF$ 
is $\mathcal{HF}$, and the axioms for it 
ensure (as far as it is possible on
the first-order level) that $HF$ is indeed to be interpreted as this
collection. In particular, we have:
\end{itemize}

\begin{lem} \cite{Cohen2014QED}
\label{lem:HF-1}The followings are provable in $RST_{HF}^m$: 
\begin{enumerate}
\item $x\in HF\leftrightarrow x=\emptyset\vee\exists u,v\in HF.u\cup\setfin v=x$. 
\item  $\left(\ensuremath{\psi\left[\nicefrac{\emptyset}{x}\right]}
\wedge\forall x\forall y \left(\psi
\wedge\ensuremath{\psi\left[\nicefrac{y}{x}\right]}\imp
\ensuremath{\psi\left[\nicefrac{x\cup\;\setfin y}{x}\right]}
\right)\right)\rightarrow\forall x\in HF.\psi$,
for $\psi\safe\emptyset$. 
\item $\ensuremath{\psi\left[\nicefrac{HF}{a}\right]}
 \wedge\forall a\left(\psi\rightarrow HF\subseteq a\right)$,
for $\psi:=\forall x\left(x\in a\leftrightarrow x=\emptyset\vee\exists u,v\in a.u\cup\setfin v=x\right)$.
\end{enumerate}
\end{lem}

\newpage
\begin{rems}\leavevmode
\be
\item An important feature of $RST$-theories is that their two axioms
directly lead (and are equivalent) to the \emph{set-theoretical }$\beta$
and $\eta$ reduction rules (see \cite{avron2008framework}). 

\item
While the formal language allows the use of set terms, it also provides
a mechanizable static check of their validity due to the syntactic
safety relation. To obtain decidable syntax, 
logically equivalent formulas are not taken to be safe w.r.t. the
same set of variables. However, if $\varphi\leftrightarrow\psi$ is provable in 
some $RST$-theory, then so is $x\in\setin x{\varphi}\leftrightarrow\psi$. 
Therefore for such $\varphi,\psi$ we might freely write in what 
follows $\setin x{\psi}$ instead of $\setin x{\varphi}$.\footnote{Further discussion on  decidability issues for safety-based languages can be found in \cite{avron2008constructibility}.} 

\item
It is easy to verify that the system $RST_{HF}^m$ is a proper
subsystem of $ZF$. While the latter is not an $RST$-theory,
in \cite{avron2008framework}
it was shown that it can be obtained from the former by
adding the following clauses to the definition of its safety relation: 
\begin{itemize}
\item {\em Separation:}
$\varphi\succ\emptyset$ for {every} formula $\varphi$.
\item {\em Powerset:}
$x\subseteq t\succ\{x\}$ if $x\not\in Fv(t)$.
\item {\em Replacement:}
 $\exists y \varphi\wedge\forall y(\varphi\rightarrow \psi)\succ X$, 
provided $\psi\succ X$,  and $X\cap Fv(\varphi)=\emptyset$.
\end{itemize}

\item
Unlike in this paper, in general
the framework for set theories just reviewed is not confined 
to the first-order level or to classical logic. Thus in
\cite{avron2010new} it was used together with ancestral logic 
(\cite{Martin43,Myhill52,shapiro1991foundations,AvronTC03,cohen2015middle}).
(This involves adding a special clause to the definition of
$\succ$ that treats the operation of transitive closure.)
Intuitionistic versions have been investigated too. 
\item A safety relation like $\succ$ presents a difficult challenge for mechanized
logical frameworks of the Edinburgh LF's type (\cite{HHP93}). First,
it is a strictly   syntactic {\em relation} between formulas and variables, whose
direct implementation requires the use of meta-variables for the variables of 
the object language --- something which is particularly difficult
to handle in this type of logical frameworks (\cite{AHMP}).  Second,
$\succ$ does not have a fixed arity like all judgements in 
the Edinburgh LF do: it is actually a relation
between formulas and {\em finite sets} of object-level  variables. Therefore
it seems that current logical frameworks should be significantly extended and refined in
order to be able to handle the syntactic framework for set theories
that was proposed in  \cite{avron2008framework} (and is used here).
\ee
\end{rems}

\subsection{The Minimal Model}

We next recall the definition of rudimentary functions (for more on
this topic see \cite{devlin1984constructibility,hrbacek1999introduction}).\footnote{To be precise, the definition we take here is given in The Basis Lemma
in \cite{devlin1984constructibility}. It was shown there that this
definition is equivalent to the standard definition of rudimentary
functions.} Rudimentary functions are just the functions obtained by omitting
the recursion schema from the usual list of schemata for primitive
recursive set functions.
\begin{defi}
\label{Def:rudimentary-function}Every {\em rudimentary function}
 is a composition
of the following functions:

\begin{itemize}
\item $F_{0}\left(x,y\right)=\left\{ x,y\right\} $
\item $F_{1}\left(x,y\right)=x-y$
\item $F_{2}\left(x,y\right)=x\times y$
\item $F_{3}\left(x,y\right)=\left\{ \left\langle u,z,v\right\rangle \,|\,z\in x\wedge\left\langle u,v\right\rangle \in y\right\} $
\item $F_{4}\left(x,y\right)=\left\{ \left\langle z,v,u\right\rangle \,|\,z\in x\wedge\left\langle u,v\right\rangle \in y\right\} $
\item $F_{5}\left(x,y\right)=\left\{ Im\left(x|_{z}\right)\,|\,z\in y\right\} $
where $Im\left(x|_{z}\right)=\left\{ w\,|\,\exists u\in z.\left\langle u,w\right\rangle \in x\right\} $
\item $F_{6}\left(x\right)=\underset{z\in x}{\bigcup}z$
\item $F_{7}\left(x\right)=Dom\left(x\right)=\left\{ v\,|\,\exists w.\left\langle v,w\right\rangle \in x\right\} $
\item $F_{8}\left(x\right)=\left\{ \left\langle u,v\right\rangle \,|\,u\in x\wedge v\in x\wedge u\in v\right\} $
\end{itemize}
\end{defi}
\medskip

\begin{defi}
~
\begin{enumerate}
\item A function is called {\em $HF$-rudimentary} if it can be generated by composition
of the functions $F_{0},...,F_{8}$ in Definition \ref{Def:rudimentary-function},
and the following constant function:

\begin{itemize}
\item $F_{9}\left(x\right)=\mathcal{HF}$ (the set of 
hereditary finite sets).
\end{itemize}
\item An {\em $HF$-universe}  ({\em universe} in short)
 is a transitive collection of sets that is closed under
$HF$-rudimentary functions.
\end{enumerate}
\end{defi}
\noindent \emph{Terminology.} In what follows, we do not distinguish
between a universe $W$ and the structure for $\lng$ with domain
$W$ and an interpretation function $I$ that assigns the obvious
interpretations to the symbols $\in$, $=$, and $\mathcal{HF}$ to $HF$.
\begin{nota}
We denote by $v\left[x:=a\right]$ the $x$-variant of $v$ which
assigns $a$ to $x$. If $\vec{y},\overrightarrow{a}$ are two vectors
of the same length we abbreviate $v\left[y_{1}:=a_{1},...,y_{n}:=a_{n}\right]$
by $v\left[\vec{y}:=\overrightarrow{a}\right]$. We denote by $\left[x_{1}:=a_{1},...,x_{n}:=a_{n}\right]$
any assignment which assigns to each $x_{i}$ the element $a_{i}$.\footnote{As long as we apply $\left[x_{1}:=a_{1},...,x_{n}:=a_{n}\right]$
to expressions whose set of free variables is contained in $\left\{ x_{1},...,x_{n}\right\} $
the exact assignment does not matter.}
\end{nota}

\begin{defi}
\label{Norm}Let $W$ be a universe, $v$ an assignment in $W$.
For any term $t$ and formula $\varphi$ of $\lng$, we recursively
define a collection $\inwv t$ and a truth value 
$\inwv{\varphi}\in\left\{ \mathbf{t},{\bf f}\right\} $
(respectively) by:

\begin{itemize}
\item $\inwv x=v\left(x\right)$ for $x$ a variable.
\item $\inwv{HF}=\mathcal{HF}$
\item $\inwv{\setin x\varphi}=\setout{a\in W}{\inwa{\varphi}^{W}={\bf t}}$
\item $\inw v{t=s}={\bf t}$ iff $\inwv t=\inwv s$ ~;~ $\inwv{t\in s}={\bf t}$
iff $\inwv t\in\inwv s$
\item $\inwv{\neg\varphi}={\bf t}$ iff $\inwv{\varphi}={\bf f}$
\item $\inwv{\varphi\wedge\psi}={\bf t}$ iff $\inwv{\varphi}={\bf t}\wedge\inwv{\psi}={\bf t}$
\item $\inwv{\varphi\vee\psi}={\bf t}$ iff $\inwv{\varphi}={\bf t}\vee\inwv{\psi}={\bf t}$
\item $\inwv{\exists x\varphi}={\bf t}$ iff $\exists a\left(a\in W\wedge\inwa{\varphi}^{W}={\bf t}\right)$
\end{itemize}
Given $W$ and $v$, we say that the term $t$ {\em defines} the collection $\inwv t$.
\end{defi}

\begin{rem}
From Theorem \ref{prop:termdefinesset} below it follows that $\inwv t$
is an element of $W$ (and it denotes the value in $W$ that the term
$t$ gets under $v$), and $\inwv{\varphi}$ denotes the truth value
of the formula $\varphi$ under $W$ and $v$.
\end{rem}
\begin{nota}
\noindent In case $exp$ is a closed term or a closed formula, we
denote by $\left\Vert exp\right\Vert ^{W}$ the value of $exp$ in
$W$, and at times we omit the superscript $W$ and simply write $\left\Vert exp\right\Vert $.
\end{nota}

\noindent The following theorem is a slight generalization of a theorem
proved in \cite{avron2010new}.

\begin{thm}
\label{thm:Arnon-1-1}~
\begin{enumerate}
\item If $F$ is an $n$-ary $HF$-rudimentary function, then there is
a formula $\varphi_{F}$ of $\mathcal{L}_{RST}^{\{HF\}}$ s.t.: 

\begin{itemize}
\item $Fv\left(\varphi_{F}\right)\subseteq\left\{ y,x_{1},...,x_{n}\right\} $
\item $\varphi_{F}\safe\left\{ y\right\} $
\item $F\left(x_{1},...,x_{n}\right)=\left\{ y\mid\varphi_{F}\right\} $
\end{itemize}
\item If $\varphi$ is a formula of $\mathcal{L}_{RST}^{\{HF\}}$ such that:

\begin{itemize}
\item $Fv\left(\varphi\right)\subseteq\left\{ y_{1},...,y_{k},x_{1},...,x_{n}\right\} $
\item $\varphi\safe\left\{ y_{1},...,y_{k}\right\} $
\end{itemize}
then there exists a $HF$-rudimentary function $F_{\varphi}$ such
that:
\[
F_{\varphi}\left(x_{1},...,x_{n}\right)=\left\{ \left\langle y_{1},...,y_{k}\right\rangle \mid\varphi\right\} 
\]
\item If $t$ is a term of $\mathcal{L}_{RST}^{\{HF\}}$ such that $Fv\left(t\right)\subseteq\left\{ x_{1},...,x_{n}\right\} $,
then there exists a $HF$-rudimentary function $F_{t}$ such that
$F_{t}\left(x_{1},...,x_{n}\right)=t$ for every $x_{1},...,x_{n}$.
\end{enumerate}
\end{thm}


\begin{proof}
The corresponding theorem in \cite{avron2010new} establishes the
connection between $\mathcal{L}_{RST}^{\{HF\}}$ without the constant
$HF$ and rudimentary functions. Thus, the only modification required
here is the treatment of the new function in (1), and the treatment
of the constant $HF$ in (2) and (3) (which are then incorporated in
the original proof that was carried out by induction). For (1), it
is easy to verify that $\varphi_{F_{9}}:=\,y=HF$. For (2) and (3) (which
are proved by simultaneous induction on the structure of terms and
formulas), the case for the constant $HF$ is immediate from the definition
of $HF$-rudimentary functions.
\end{proof}
\begin{thm}
\label{prop:termdefinesset}Let $W$ be a universe, and $v$ an assignment
in $W$.
\end{thm}
\begin{itemize}
\item For $t$ a term of $\lng$, $\inwv t\in W$.
\item For $\varphi$ a formula of $\lng$:

\begin{itemize}
\item If $\varphi\safe\left\{ y_{1},...,y_{n}\right\} $ and $n>0$, then:
\[
\setout{\left\langle a_{1},...,a_{n}\right\rangle \in W^{n}}{\left\Vert \varphi\right\Vert _{v\left[\vec{y}:=\overrightarrow{a}\right]}^{W}={\bf t}}\in W
\]
.
\item If $\varphi\safe\emptyset$ and $\left\{ y_{1},...,y_{n}\right\} \subseteq Fv\left(\varphi\right)$,
then for any $X\in J_{2}$:
\[
\setout{\left\langle a_{1},...,a_{n}\right\rangle \in X^{n}}{\left\Vert \varphi\right\Vert _{v\left[\vec{y}:=\overrightarrow{a}\right]}^{W}={\bf t}}\in W
\]
\end{itemize}
\end{itemize}

\begin{proof}
The proof is straightforward using Theorem \ref{thm:Arnon-1-1}. Claims
(1) and (2a) are immediate. For (2b) let $\varphi$ be a formula s.t.
$\varphi\safe\emptyset$ and $\left\{ y_{1},...,y_{n}\right\} \subseteq Fv\left(\varphi\right)$.
Using Theorem \ref{thm:Arnon-1-1} we get that $\varphi$ defines
a $HF$-rudimentary predicate, $P_{\varphi}$ (i.e. one whose characteristic
function is $HF$-rudimentary). Define:
\[
H\left(x_{1},...,x_{k},y_{1},...,y_{n}\right)=\begin{cases}
\left\{ \left\langle y_{1},...,y_{n}\right\rangle \right\}  & \mathrm{if}\,\,\,P_{\varphi}\left(x_{1},...,x_{k},y_{1},...,y_{n}\right)\\
\emptyset & \mathrm{otherwise}
\end{cases}
\]
$H$ is a $HF$-rudimentary function (see Lemma 1.1 in \cite{devlin1984constructibility}).
Now, define:
\[
F\left(z,x_{1},...,x_{k}\right)=z^{n}\cap\left\{ \left\langle y_{1},...,y_{n}\right\rangle |P_{\varphi}\left(x_{1},...,x_{k},y_{1},...,y_{n}\right)\right\} 
\]
$F$ is also $HF$-rudimentary since
\[
F\left(z,x_{1},...,x_{k}\right)=\bigcup_{\left\langle y_{1},...,y_{n}\right\rangle \in z^{n}}H\left(x_{1},...,x_{k},y_{1},...,y_{n}\right)
\]
Now, the fact that $W$ is a universe entails that for every assignment
$v$ in $W$ and every 
$X\in W$, $F\left(X,v\left(x_{1}\right),....,v\left(x_{k}\right)\right)\in W$,
and so
$\setout{\left\langle a_{1},...,a_{n}\right\rangle \in X^{n}}{\left\Vert \varphi\right\Vert _{v\left[\vec{y}:=\overrightarrow{a}\right]}^{W}={\bf t}}\in W$.
\end{proof}

\begin{prop}
\label{prop:TranRud} Let $W$ be a universe. Then, $W$ is a model
of $RST_{HF}^m$.
\end{prop}
\begin{proof}
The Extensionality axiom is clearly satisfied in any universe. Theorem
\ref{prop:termdefinesset} entails that the interpretation of any
term is an element of the universe. This immediately implies
that the other axioms are satisfied in any universe. It is also straightforward
to verify that the interpretation of $HF$ as $\mathcal{HF}$ satisfies
the three axioms for $HF$.
\end{proof}



\subsection{\label{sec:computational} Computational Theories and Universes}

\
\ms

Computations within a set of objects require concrete representations
of these objects. Accordingly,
we call a theory \emph{computational}
if its set of closed terms induces in a natural way a minimal model of the theory,
and it enables the key properties of these elements to be provable
within it. Next we provide a  more formal definition
for the case of  set theories which are defined within our general
framework.  Note that from  a Platonist point of view, the set of closed terms of such a theory
$\mathcal{T}$  induces some subset $\mathcal{ S_{T}}$
of the cumulative universe of sets $V$, 
 as well as some subset $\mathcal{ M_{T}}$ of any transitive
model $\mathcal{ M}$ of $\mathcal{T}$. 

\begin{defi}
\label{def: computability}
\ 
\begin{enumerate}
\item
A theory $\mathcal{T}$ in the above framework is called
\emph{computational} if the set $\mathcal{ S_{T}}$ it induces is a transitive 
model of $\mathcal{T}$, and the identity of $\mathcal{ S_{T}}$
is absolute in the sense that $\mathcal{ M_{T}}=\mathcal{ S_{T}}$ for any
transitive model $\mathcal{ M}$ of $\mathcal{T}$ (implying that $\mathcal{ S_{T}}$
is a \emph{minimal} transitive model of $\mathcal{T}$). 
\item A set is called {\em computational} if it is $\mathcal{ S_{T}}$ for some
computational theory $\mathcal{T}$.
\end{enumerate}
\end{defi}

The most basic computational theories are $RST^m$ and $RST_{HF}^m$,
which are the two minimal theories
in the hierarchy of systems developed in \cite{Cohen2014QED}.
This fact, as well as the corresponding computational universes, are described
in the following three results from \cite{Cohen2014QED}.

\begin{prop}
\label{prop:TranRud-1}Let $J_{1}$ and $J_{2}$ be the first two elements
in Jensen's hierarchy \cite{Jensen1972fine}.\footnote{$J_{1}=\mathcal{HF}$ and 
$J_{2}=Rud\left(J_1\right)$,
where $Rud\left(x\right)$ denotes the smallest set $y$ such that
$x\subseteq y$, $x\in y$, and $y$ is closed under application of
all rudimentary functions.}
\begin{enumerate}
\item $J_{1}$ is a model of $RST$.
\item $J_{2}$ with the interpretation of $HF$ as $J_{1}$ is a model of
$RST_{HF}^m$. 
\end{enumerate}
\end{prop}

\begin{proof}
The first claim is trivial. The second claim follows from Corollary
\ref{prop:TranRud}, since $J_{2}$ is clearly a universe.
\end{proof}

\begin{thm}
~\label{thm:in J2} 
\be
\item $X\in J_{1}$ iff there is a closed term $t$ of $\mathcal{L}_{RST}$
s.t. $\left\Vert t\right\Vert ^{J_{1}}=X$.
\item $X\in J_{2}$ iff there is a closed term $t$ of $\lng$ such that
$\left\Vert t\right\Vert ^{J_{2}}=X$.
\ee
\end{thm}

\begin{proof}
We prove the second item, leaving the easy proof of the first to the reader. 
Theorem \ref{prop:termdefinesset} entails the right-to-left implication.
The converse is proved by induction, using Lemma~\ref{lem:example}. 
Clearly, $\left\Vert \setin x{x\in x}\right\Vert ^{J_{2}}=\emptyset$
and $\left\Vert HF\right\Vert ^{J_{2}}=J_{1}$. Now, suppose that
for $A,B\in J_{2}$ there are closed terms $t_{A}$ and $t_{B}$ such
that $\left\Vert t_{A}\right\Vert ^{J_{2}}=A$ and $\left\Vert t_{B}\right\Vert ^{J_{2}}=B$.
We show that there are closed terms for any of the results of applications
of $F_{0},...,F_{8}$ to $A$ and $B$.

\begin{itemize}
\item $F_{0}\left(A,B\right)=\left\Vert \setfin{t_{A},t_{B}}\right\Vert ^{J_{2}}$
\item $F_{1}\left(A,B\right)=\left\Vert t_{A}-t_{B}\right\Vert ^{J_{2}}$
\item $F_{2}\left(A,B\right)=\left\Vert t_{A}\times t_{B}\right\Vert ^{J_{2}}$
\item $F_{3}\left(A,B\right)=\left\Vert \setin x{\exists z\in t_{A}\exists u,v.\left\langle u,v\right\rangle \check{\in}t_{B}\wedge x=\left\langle u,z,v\right\rangle }\right\Vert ^{J_{2}}$
\item $F_{4}\left(A,B\right)=\left\Vert \setin x{\exists z\in t_{A}\exists u,v.\left\langle u,v\right\rangle \check{\in}t_{B}\wedge x=\left\langle z,v,u\right\rangle }\right\Vert ^{J_{2}}$
\item $F_{5}\left(A,B\right)=\left\Vert \setin{Im\left(\setin w{w\in t_{A}\wedge\pi_{1}\left(w\right)\in z}\right)}{z\in t_{B}}\right\Vert ^{J_{2}}$
\item $F_{6}\left(A\right)=\left\Vert \setin x{\exists u\in t_{A}.x\in u}\right\Vert ^{J_{2}}$
\item $F_{7}\left(A\right)=\left\Vert Dom\left(t_{A}\right)\right\Vert ^{J_{2}}$
\item $F_{8}\left(A\right)=\left\Vert \setin x{\exists u\in t_{A}\exists v\in t_{a}.u\in v\wedge x\check{=}\left\langle u,v\right\rangle }\right\Vert ^{J_{2}}$
\qedhere
\end{itemize}
\end{proof}

\begin{cor}
$RST^m$ and $RST_{HF}^m$ are computational, and $J_{1}$
and $J_{2}$ are their computational models.
\end{cor}

Now $J_{1}$, the minimal computational set, is the set of hereditary finite sets.
Its use captures the
standard data structures used in computer science, like strings
and lists. However, in order to be able to capture computational structures
with infinite objects, we have to move to $RST_{HF}^m$,
whose computational universe, $J_{2}$,
seems to be the minimal universe that suffices for this purpose. 
$RST_{HF}^m$ still allows 
for a very concrete, computationally-oriented interpretation, and
it is appropriate for mechanical manipulations and interactive
theorem proving. As noted in the introduction, the main goal of this paper
is to show that this theory and its corresponding  
universe $J_2$ are sufficiently rich for a systematic development of
(great parts of) applicable mathematics.


\section{\label{sec:Static-Extensions-by}Classes and Static Extensions by Definitions}

When working in a minimal computational universe such as $J_{2}$
(as done in the next section), many of the standard mathematical objects
(such as the real line and real functions) are only available in our
framework as proper classes. Thus, in order to be able to formalize
standard theorems regarding such objects we must enrich our language
to include them. However, introducing classes into our framework
is a part of the more general method of extensions by definitions,
which is an essential part of every mathematical research and its
presentation. There are two principles that govern this process
in our framework. First, the static nature of our framework demands
that conservatively expanding the language of a given theory should
be reduced to the use of \emph{abbreviations}. Second, since the introduction
of new predicates and function symbols creates new atomic formulas
and terms, one should be careful that the basic conditions concerning
the underlying safety relation $\succ$ are preserved. Thus only formulas
$\varphi$ s.t. $\varphi\succ\emptyset$ can be used for defining
new predicate symbols. 

We start with the problem of introducing new predicate symbols. Since
$n$-ary predicates can be reduced in the framework of set theory
to unary predicates, we focus on the introduction of new unary
predicates.  In standard practice such extensions are carried out by introducing
a new unary predicate symbol $P$ and either treating $P\left(t\right)$
as an abbreviation for $\varphi\left[ \nicefrac{t}{x}\right] $
 for some formula $\varphi$ and variable $x$,
or (what is more practical) adding $\forall x\left(P\left(x\right)\leftrightarrow\varphi\right)$
as an axiom to the (current version of the base) theory, obtaining
by this a conservative theory in the extended language. However, in
the set theoretical framework it is possible and frequently more convenient
to uniformly use class terms, rather than introduce a new predicate
symbol each time.  
Thus, instead of writing ``$P\left(t\right)$'' one uses an appropriate
class term $S$ and writes ``$t\in S$''. Whatever approach is chosen
\textendash{} in order to respect the definition of a safety relation,
class terms should be restricted so that ``$t\in S$'' is safe w.r.t.
$\emptyset$. Accordingly, we extend our language by incorporating
class terms, which are objects of the form $\class x{\varphi}$, where
$\varphi\saferst\emptyset$. The use of these terms is done in the
standard way. 
  In particular, $t\in\class x{\varphi}$
(where $t$ is free for $x$ in $\varphi$) is equivalent to (and
may be taken as an abbreviation for) $\varphi\left[ \nicefrac{t}{x}\right] $.
It should be emphasized that a class term is not a valid  term
in the language, but only a definable predicate. Thus the  addition
of the new notation does not enhance the expressive power of 
languages like \emph{$\lng$},
but only increases the ease of using them.

\begin{rem}
{\rm
Further standard abbreviations
(see \cite{levy1979basic}) are:
\begin{itemize}
\item $t\subseteq\class x{\varphi}$ is an abbreviation for $\forall z\left(z\in t\rightarrow z\in\class x{\varphi}\right)$.
\item $t=\class x{\varphi}$ and $\class x{\varphi}=t$ stand for $\forall z\left(z\in t\leftrightarrow z\in\class x{\varphi}\right)$.
\item $\class x{\varphi}=\class y{\psi}$ is an abbreviation for $\forall z\left(z\in\class x{\varphi}\leftrightarrow z\in\class y{\psi}\right)$.
\item $\class x\varphi\in t$ is an abbreviation for $\exists z.z=\class x\varphi\wedge z\in t$.
\item $\class x\varphi\in\class y\psi$ is an abbreviation for $\exists z.z=\class x\varphi\wedge z\in\class y\psi$.
\end{itemize}
Note that these formulas are merely abbreviations for formulas which
are \emph{not} necessarily atomic (even though, $t\subseteq\class x\varphi$
also happens to be safe with respect to $\emptyset$).
}\end{rem}

A further conservative extension of the language that we shall use
incorporates free class variables, $\boldsymbol{X},\boldsymbol{Y},\boldsymbol{Z}$,
and free function variables, $\boldsymbol{F},\boldsymbol{G}$, into
$\lng$ (as in free-variable second-order logic  \cite{shapiro1991foundations}).
These variables stand for arbitrary class or function terms (the latter
is defined in Def. \ref{def:function}), and they may only appear
as \emph{free} variables, \emph{never to be quantified}. We allow
occurrences of such variables inside a formula in a class term or
a function term. One may think of a formula with such variables as
a schema, where the variables play the role of ``place holders'',
and whose substitution instances abbreviate official formulas of the
language. (See Example \ref{Example2}.) In effect, a formula $\psi\left(\boldsymbol{X}\right)$
with free class variable $\boldsymbol{X}$ can be intuitively interpreted
as ``for any \emph{given} class $X$, $\psi\left(X\right)$ holds''.
Thus, a free-variable formulation has the flavor of a universal formula.
Therefore, this addition allows statements about \emph{all
}potential classes and \emph{all }potential functions. 

\begin{defi}
\label{Norm-class}
Let $W$ be a universe, $v$ an assignment in $W$, and let $\varphi\succ\emptyset$. Define:
$$\inwv{\class x{\varphi}}:=\setout{a\in W}{\left\Vert \varphi\right\Vert _{v\left[x:=a\right]}^{W}={\bf t}}$$
Given $W$ and $v$, we again say that the class term on the left {\em defines}
here the collection on the right (even though it might
not be an element of $W$).
\end{defi}

\begin{defi}
\label{Def class}Let $X$ be a collection of elements in a universe $W$. 
\begin{itemize}
\item $X$ is a $\i$\emph{-set} (in $W$) if there is a closed term that defines
it. (See Definition~\ref{Norm}.)  If $X$ is a \emph{$\i$}-set, $\widetilde{X}$ denotes some
closed term that defines it.
\item $X$ is a \emph{$\i$-class} (in $W$) if there is a closed class term that
defines it. If $X$ is a \emph{$\i$}-class, $\bar{X}$ denotes some
closed class term that defines it.
\end{itemize}
\end{defi}
Note that by Corollary \ref{prop:termdefinesset}, if $X$ is a $\i$\emph{-set} in $W$
then $X\in W$. 
\begin{prop}
\label{prop:Basic setclass}The following holds for every universe $W$:  
\begin{enumerate}
\item Every $\i$-set is a \emph{$\i$}-class. 
\item The intersection of a \emph{$\i$}-class with a \emph{$\i$}-set is
a \emph{$\i$}-set. 
\item Every $\i$-class that is contained in a \emph{$\i$}-set is a \emph{$\i$}-set.
\end{enumerate}
\end{prop}
\begin{proof}\leavevmode
\begin{enumerate}
\item If $X$ is a \emph{$\i$}-set, then $x\in\widetilde{X}\succ\{x\}$.
Hence (see \cite{avron2008framework}) $x\in\widetilde{X}\safe\emptyset$. 
This implies that $\class x{x\in\widetilde{X}}$
is a class term which defines $X$, and so $X$ is a \emph{$\i$}-class.
\item Let $X$ be a \emph{$\i$}-class and $Y$ be a \emph{$\i$}-set. Then,
$X\cap Y$ can be defined by the term $\setin z{z\in\bar{X}\wedge z\in\widetilde{Y}}$.
Since $z\in\bar{X}\safe\emptyset$ and $z\in\widetilde{Y}\safe\left\{ z\right\} $
we get that $z\in\bar{X}\wedge z\in\widetilde{Y}\safe\left\{ z\right\} $.
Hence $X\cap Y$ is a \emph{$\i$}-set. 
\item Follows from (2), since if $X\subseteq Y$ then $X=X\cap Y$.  \qedhere
\end{enumerate}
\end{proof}

\begin{rem}
\label{note:class} {\rm A semantic counterpart of our notion of a $\i$-class
was used in \cite{weaver2005analysis}, and is called there an
$\iota$-class. It is defined  as a definable subset of $J_{2}$
whose intersection with any element of $J_{2}$ is in $J_{2}$.
The second condition in this definition seems somewhat
ad hoc. More importantly, it is  unclear how it can be
checked in general, and what kind of set theory is needed to establish
 that certain collections are $\iota$-classes. In contrast, the definition
of a $\i$-class used here is motivated by, and based
on, purely syntactical considerations. It is also  a simplification of
the notion of $\iota$-class, as by Prop. \ref{prop:Basic setclass}(2)
every $\i$-class is an $\iota$-class.}\footnote{Two other ideas that appear in the sequel were adopted from \cite{weaver2005analysis}:
treating the collection of reals as a proper class, and the use of
codes for handling certain classes. It should nevertheless be emphasized
that the framework in \cite{weaver2005analysis} is exclusively based
on semantical considerations, and it is unclear how it can be turned
into a formal (and suitable for mechanization) theory like $ZF$ or $PA$.} 
\end{rem}
\begin{prop}
\label{Pro:SetPairs-1}The following holds for every universe $W$: 
\begin{itemize}
\item Let $Y$ be a \emph{$\i$}-set. If $\varphi\safe\emptyset$ and $Fv\left(\varphi\right)\subseteq\left\{ x\right\} $,
then $\left\{ x\in Y\mid\varphi\right\} $ is a \emph{$\i$}-set. 
\item If $\varphi\safe\left\{ x_{1},...,x_{n}\right\} $, then $\left\{ \left\langle x_{1},...,x_{n}\right\rangle \mid\varphi\right\} $
is a \emph{$\i$}-set. 
\end{itemize}
\end{prop}
\begin{proof}\leavevmode
\begin{itemize}
\item $\left\{ x\in Y\mid\varphi\right\} $ is defined by $\setin x{x\in\widetilde{Y}\wedge\varphi}$. 
\item $\left\{ \left\langle x_{1},...,x_{n}\right\rangle \mid\varphi\right\} $
is defined by $\setin z{\exists x_{1}...\exists x_{n}\left(\varphi\wedge z=\left\langle x_{1},...,x_{n}\right\rangle \right)}$,
where $z$ is fresh. \qedhere
\end{itemize}
\end{proof}

\begin{prop}
\label{cor:Every-rudimentary-functionTERM}For every $n$-ary $HF$-rudimentary
function $f$ there is a term $t$ with $Fv\left(t\right)\subseteq\left\{ x_{1},...,x_{n}\right\} $
s.t. for any $\left\langle A_{1},...,A_{n}\right\rangle \in W^{n}$,
$f$ returns the \emph{$\i$}-set $\left\Vert t\right\Vert _{\left[x_{1}:=A_{1},...,x_{n}:=A_{n}\right]}^{W}$.
\end{prop}

\begin{proof}
It is easy to see that if $X_{1},...,X_{n}$ are $\i$-sets, and $\varphi$
is a formula such that $Fv\left(\varphi\right)\subseteq\left\{ y,v_{1},...v_{n}\right\} $
and $\varphi\safe\left\{ y\right\} $, then $\left\{ y\mid\varphi\left\{ \frac{\widetilde{X_{1}}}{v_{1}},...,\frac{\widetilde{X_{n}}}{v_{n}}\right\} \right\} $
is a \emph{$\i$}-set. Therefore the proposition easily follows from
Theorem \ref{thm:Arnon-1-1}.
\end{proof}

\begin{prop}
\label{pro:classes}If $X,Y$ are $\i$-classes (in a universe $W$), so are $X\cup Y$,
$X\cap Y$, $X\times Y$, $J_{2}-X$, and $P_{J_{2}}\left(X\right)=\left\{ z\in J_{2}\mid z\subseteq X\right\} $.
\end{prop}

\noindent
{\em Proof:}
\hfill
\begin{itemize}
\item $\overline{X\cup Y}=\class x{x\in\bar{X}\vee x\in\bar{Y}}$. 
\item $\overline{X\cap Y}=\class x{x\in\bar{X}\wedge x\in\bar{Y}}$. 
\item $\overline{X\times Y}=\class x{\exists a\exists b\left(a\in\bar{X}\wedge b\in\bar{Y}\wedge x\check{=}\left\langle a,b\right\rangle \right)}$.
(See Lemma~\ref{prop:pairs}.)
\item $\overline{J_{2}-X}=\class x{x\notin\bar{X}}$. 
\item $P_{J_{2}}\left(X\right)=\class z{z\subseteq\bar{X}}
=\class z{\forall a\left(a\in z\rightarrow a\in\bar{X}\right)}$. 
(See footnote~\ref{foot:forall}.) \qed
\end{itemize}

\ms

For a class term $s$ we denote by $2^{s}$ the class term $\class z{z\subseteq s}$.
Note that for any assignment $v$ in $W$ and class term $s$, $\inwv{2^{s}}$
is equal to $P_{W}\left(\inwv s\right)$, i.e., the intersection of
the power set of $\inwv s$ and $W$. This demonstrates the main
difference between set terms and class terms. The interpretation of
set terms is absolute, whereas the interpretation of class terms might
not be (though membership in the interpretation of a class term {\em is}
absolute). 
\begin{defi}
Let $W$ be a universe.
A \emph{$\i$}-\emph{relation (in $W$) from }a \emph{$\i$}-\emph{class $X$
to }a \emph{$\i$}-\emph{class $Y$ }is a \emph{$\i$}-class $A$
s.t. $A\subseteq X\times Y$. A $\i$-relation is called\emph{
small }if it is a \emph{$\i$}-set (of $W$).
\end{defi}


Next we extend our framework by the introduction of new function symbols.
This poses a new difficulty. While new relation symbols are commonly
introduced in a static way, new function symbols are usually introduced
\emph{dynamically}: a new function symbol is made available after
appropriate existence and uniqueness theorems had been proved. 
 However, one of the main
guiding principles of our framework is that its languages should be
treated exclusively in a \emph{static} way. Thus function symbols,
too, are introduced here only as abbreviations for definable operations
on sets.\footnote{In this paper, as in standard mathematical
textbooks, the term ``function'' is used both for collections of
ordered pairs and for set-theoretical operations (such as $\cup$). }
\begin{defi}
\label{def:function} 
Let $W$ be a universe. (The various definitions should be taken with respect to $W$.)
\begin{itemize}
\item For a closed class term  $\mathsf{T}$ and a term $t$ of $\lng$, $\lambda x\in\mathsf{T}.t$
is a \emph{function term} which is an abbreviation for 
$\class z{\exists x\exists y\left(z\check{=}\left\langle x,y\right\rangle 
\wedge x\in\mathsf{T}\wedge y=t\right)}$.\footnote{We abbreviate by 
$z\check{=}\left\langle x,y\right\rangle $ and 
$\left\langle x,y\right\rangle \check{\in}z$ the two formulas that are provably
equivalent to $z=\left\langle x,y\right\rangle $ and $\left\langle x,y\right\rangle \in z$
and are safe w.r.t. $\left\{ x,y\right\} $ which were introduced in \cite{Cohen2014QED}.}
\item A \emph{$\i$}-class $F$ is a \emph{$\i$}-\emph{function
on a} \emph{$\i$-class $X$} if there is a function term $\lambda x\in\mathsf{T}.t$
such that $X=\left\Vert\mathsf{T}\right\Vert$, $Fv\left(t\right)\subseteq\left\{ x\right\} $ and $F=\left\Vert \lambda x\in\mathsf{T}.t\right\Vert $.
$t$ is called a term which \emph{represents} $F$.
\item A $\i$-class is called a $\i$-function if it is a $\i$-function
on some $\i$-class.
\item A $\i$-function is called\emph{ small
}if it is a \emph{$\i$}-set.
\end{itemize}
\end{defi}
\noindent Note that the standard functionality condition is always
satisfied by a \emph{$\i$}-\emph{function}.

\ms

\noindent \emph{Terminology. }In what follows, claiming that an object
is \emph{available in $RST_{HF}^m$ as a $\i$-function ($\i$-relation)}
means that for every universe $W$, the object  is definable in $\lng$
as a $\i$-function \emph{($\i$-relation)} of $W$,
and that its basic properties are provable in $RST_{HF}^m$.\footnote{The ``basic properties'' of a certain object is of course a fuzzy
notion. However, it is not difficult to identify its meaning in each
particular case, as will be demonstrated in several examples below.} 
\begin{prop}
\label{prop:relation}Let $X,Y$ be $\i$-classes and $R$ a \emph{$\i$}-relation
from $X$ to $Y$. 
\begin{enumerate}
\item $R$ is small iff $Dom\left(R\right)$ and $Im\left(R\right)$ are
$\i$-sets. 
\item $R^{-1}=\left\{ \left\langle y,x\right\rangle |\left\langle x,y\right\rangle \in R\right\} $
is available in $RST_{HF}^m$ as a \emph{$\i$}-relation from $Y$
to $X$. If $R$ is small, then so is $R^{-1}$. 
\item If $Z\subseteq X$ and $U\subseteq Y$ are \emph{$\i$}-classes, then
$R\cap\left(Z\times U\right)$ is available in $RST_{HF}^m$ as
a \emph{$\i$}-relation from $Z$ to $U$ . 
\end{enumerate}
\end{prop}
\begin{proof}\leavevmode
\begin{enumerate}
\item $(\Rightarrow)$ If $R$ is a $\i$-set, then $\exists y.\left\langle x,y\right\rangle \check{\in}\widetilde{R}\safe\left\{ x\right\} $
and $\exists x.\left\langle x,y\right\rangle \check{\in}\widetilde{R}\safe\left\{ y\right\} $.
Thus, $Dom\left(R\right)$ is defined by $\setin x{\exists y.\left\langle x,y\right\rangle \check{\in}\widetilde{R}}$
and $Im\left(R\right)$ by $\setin y{\exists x.\left\langle x,y\right\rangle \check{\in}\widetilde{R}}$.\\
 $(\Leftarrow)$ If $Dom\left(R\right)$ and $Im\left(R\right)$ are
$\i$-sets, then $Dom\left(R\right)\times Im\left(R\right)$ is a
$\i$-set as $\times$ is a rudimentary function. Since $R$ is a
$\i$-class such that $R\subseteq Dom\left(R\right)\times Im\left(R\right)$,
Prop. \ref{prop:Basic setclass} entails that $R$ is a $\i$-set.
\item Since $R$ is a $\i$-class, we can define $\overline{R^{-1}}=\class z{\exists x\exists y\left(\left\langle x,y\right\rangle \in\bar{R}\wedge z\check{=}\left\langle y,x\right\rangle \right)}$.
Now, $\exists x\exists y\left(\left\langle x,y\right\rangle \in\bar{R}\wedge z\check{=}\left\langle y,x\right\rangle \right)\safe\emptyset$,
as $z\check{=}\left\langle y,x\right\rangle \safe\left\{ x,y\right\} $
and $\left\langle x,y\right\rangle \in\bar{R}\safe\emptyset$. It
is standard to prove in $RST_{HF}^m$ properties such as $\left\langle x,y\right\rangle \in R\leftrightarrow\left\langle y,x\right\rangle \in R^{-1}$
and $\left(R^{-1}\right)^{-1}=R$. If $R$ is a $\i$-set, $R^{-1}$
can be defined by $\setin z{\exists x\exists y\left(\left\langle x,y\right\rangle \in\widetilde{R}\wedge z=\left\langle y,x\right\rangle \right)}$,
hence $R^{-1}$ is a $\i$-set. 
\item Surely $R\cap\left(Z\times U\right)\subseteq Z\times U.$ By Prop.\ref{pro:classes},
since $R,Z,U$ are $\i$-classes, we have that $R\cap\left(Z\times U\right)$
is a $\i$-class. \qedhere
\end{enumerate}
\end{proof}

\begin{prop}
\label{prop:standard func}A $\i$-set is a function according to
the standard mathematical definition (a single-valued relation) iff it is a small $\safe$-function. 
\end{prop}

\begin{proof}
Let $A$ be a $\i$-set which is a relation that satisfies the functionality
condition. Since $A$ is a \emph{$\i$}-set, there is a closed term
$\widetilde{A}$ that defines it. $A$ is a $\i$\emph{-}function
on the\emph{ }$\i$-set $Dom\left(A\right)$ since the term $t=\iota y.\left\langle x,y\right\rangle \check{\in}\widetilde{A}$
represents it (see Lemma \ref{lem:example}). The term $t$ is legal
and it represents $A$ since $\left\langle x,y\right\rangle \check{\in}\widetilde{A}\safe\left\{ y\right\} $
and $A$ satisfies the functionality condition. The converse is trivial,
since for every small $\safe$-function there is a term representing
it, and thus the functionality condition clearly holds by the equality
axioms of $FOL$.
\end{proof}

\noindent \emph{Notation}. Let $F=\left\Vert \lambda x\in\bar{X}.t\right\Vert $
be a \emph{$\i$}-function.  We employ standard $\beta$-reduction
for $\lambda$ terms. Thus, we write $F\left(s\right)$ for $t\left[ \nicefrac{s}{x}\right] $ if $s$ is  free for $x$ in
$t$.
Hence $F\left(s\right)=y$ stands for $t\left [ \nicefrac{s}{x}\right] =y$,
and so if $y\notin Fv\left[t\right]\cup Fv\left[s\right]\setminus\left\{ x\right\} $,
then $F\left(s\right)=y\safe\left\{ y\right\} $. 
\begin{prop}
[Replacement axiom in class form] \label{prop:replacement}
Let $F$ be a \emph{$\i$}-function on a \emph{$\i$}-class $X$.
Then for every $\i$-set $A\subseteq X$, $F\left[A\right]=\left\{ F\left(a\right)\,|\,a\in A\right\} $
is a \emph{$\i$}-set.
\end{prop}

\begin{proof}
The term $\setin y{\exists a\in\widetilde{A}.F\left(a\right)=y}$
defines $F\left[A\right]$. 
\end{proof}

\noindent Below is a natural generalization of Def. \ref{def:function}
to functions of several variables. 
\begin{lem}
\label{lem:Functions Severl variables}If \textup{$X_{1},...,X_{n}$
}\emph{are $\i$}-classes and $t$ is a term s.t. 
$Fv\left(t\right)\subseteq\left\{ x_{1},...,x_{n}\right\} $,
then $F=\left\Vert \lambda x_{1}\in\bar{X_{1}},...,x_{n}\in\bar{X_{n}}.t\right\Vert $
is available in $RST_{HF}^m$ as a\emph{ }$\i$\emph{-}function
on\emph{ }\textup{$X_{1}\times...\times X_{n}$. }\textup{\emph{(Here}}
$\lambda x_{1}\in\bar{X_{1}},...,x_{n}\in\bar{X_{n}}.t$ abbreviates
$\class{\left\langle \left\langle x_{1},...,x_{n}\right\rangle ,t\right\rangle }{\left\langle x_{1},...,x_{n}\right\rangle \in\bar{X_{1}}\times...\times\bar{X_{n}}}$\textup{\emph{).}}
\end{lem}
\begin{cor}
Every $HF$-rudimentary function is available in $RST_{HF}^m$ as
a \emph{$\i$}-function.
\end{cor}
\begin{prop}
\label{prop:functions}Let $F$ be a \emph{$\i$}-function on a \emph{$\i$}-class
$X$. 
\begin{enumerate}
\item \label{enu:-is-small}$F$ is small iff $X$ is a \emph{$\i$}-set. 
\item If $Y_{0}$ is a \emph{$\i$}-class, then $F^{-1}\left[Y_{0}\right]=\left\{ a\in X\mid F\left(a\right)\in Y_{0}\right\} $
is a \emph{$\i$}-class. If $F$ is small, then $F^{-1}\left[Y_{0}\right]$
is a \emph{$\i$}-set. 
\item If $X_{0}\subseteq X$ is a \emph{$\i$}-class, then $F\upharpoonright_{X_{0}}$
is available in $RST_{HF}^m$ as a \emph{$\i$}-function. 
\item $G\circ F$ is available in $RST_{HF}^m$ as a \emph{$\i$}-function
on $X$, in case $G$ is a \emph{$\i$}-function on a \emph{$\i$}-class
$Y$ and $Im(F)\subseteq Y$.
\item If $G$ is a \emph{$\i$}-function on a \emph{$\i$}-class $Y$ and
$F$ and $G$ agree on $X\cap Y$, then $G\cup F$ is available in
$RST_{HF}^m$ as a \emph{$\i$}-function on $X\cup Y$. 
\item If $Z$ is a \emph{$\i$}-class then the identity on $Z$ and any
constant function on $Z$ are available in $RST_{HF}^m$ as \emph{$\i$}-functions.
\end{enumerate}
\end{prop}
\begin{proof}\leavevmode
\begin{enumerate}
\item $(\Rightarrow)$ If $F$ is a $\i$-set, then $Dom\left(F\right)=X$
is also a $\i$-set, since $Dom$ is  rudimentary.\\
 $(\Leftarrow)$ Suppose $t$ represents $F$. If $X$
is a $\i$-set, then $F=\left\Vert \lambda x\in\widetilde{X}.t\right\Vert $
which is a $\i$-set.
\item $\overline{F^{-1}\left[Y_{0}\right]}=\class a{a\in\bar{X}\wedge F\left(a\right)\in\bar{Y_{0}}}$
. Since $a\in\bar{X}\wedge F\left(a\right)\in\bar{Y_{0}}\safe\emptyset$,
we get that $F^{-1}\left[Y_{0}\right]$ is a $\i$-class. If $F$
is small, then by (\ref{enu:-is-small}) we have that $X$ is a $\i$-set.
The fact that $F^{-1}\left[Y_{0}\right]\subseteq X$ entails that
$F^{-1}\left[Y_{0}\right]$ is a $\i$-set. 
\item If $t$ is a term that represents $F$, it also represents $F\upharpoonright_{X_{0}}$.
\item Denote by $t_{F},t_{G}$ terms that represent the $\i$-functions
$F,G$ respectively. Thus, $G\circ F=\left\Vert \lambda x\in\overline{X}.t_{G}\left\{ \frac{t_{F}\left(x\right)}{x}\right\} \right\Vert $.
It is easy to see that standard properties, such as the associativity
of $\circ$, are provable in $RST_{HF}^m$. 
\item Denote by $t_{F},t_{G}$ terms that represent the $\i$-functions
$F,G$ respectively. Then:
$$F\cup G=\left\Vert \lambda x\in\bar{X}\cup\bar{Y}.\iota y.\left(x\in\bar{X}\wedge y=t_{F}\right)\vee\left(x\in\bar{Y}-\bar{X}\wedge y=t_{G}\right)\right\Vert$$ 
Since each of the disjuncts is safe w.r.t $\left\{ y\right\} $, we
get that the term is valid. It is easy to verify that in $RST_{HF}^m$
basic properties, such as $\forall x\in\bar{X}$.$\overline{G\cup F}\left(x\right)=\bar{F}\left(x\right)$,
are provable.
\item $id_{Z}=\left\Vert \lambda z\in\bar{Z}.z\right\Vert $, and for any
$A\in J_{2}$, $const_{A}=\left\Vert \lambda z\in\bar{Z}.\tilde{A}\right\Vert $.
Proving basic properties such as $\forall x,y\in\bar{Z}.const_{A}\left(x\right)=const_{A}\left(y\right)$
in $RST_{HF}^m$ is routine. \qedhere
\end{enumerate}
\end{proof}

\section{\label{sec:The-Natural-Numbers}The Natural Numbers}

We introduce the natural numbers by following their
standard construction in set theory:
$0:=\emptyset;\,\,\,\,\,\,n+1:=S\left(n\right)$,
where $S\left(x\right)=x\cup\left\{ x\right\} $. Obviously, 
each  $n\in\mathbb{N}$
is a \emph{$\i$}-set in any universe, and $\mathbb{N}$ (the set of natural numbers)
is contained in $\mathcal{HF}=J_1$. What is more, the property
of being a natural number is defined in any universe by the following formula:
$$ N(x) := \ 
\forall y\in x\cup\{x\}.y=\emptyset\vee\exists
w\in x.y=w\cup{ \{}w{ \}}$$
Note that this formula  has the same extension in any transitive set which includes
$\emptyset$ and is closed under the operation 
$\lambda x. x\cup\{x\}$.\footnote{This is a significant
improvement on \cite{Cohen2014QED}, in which another formula, 
$Ord\left(x\right)$,
has been used for characterizing $\mathbb{N}$
as $\{x\mid x\in HF\w Ord(x)\}$. However, $Ord\left(x\right)$ is 
actually true for all ordinals, and so it lacks
the strong absoluteness property that $N(x)$ has.}
It follows that from a semantic point of view it should be taken as
safe with respect to $x$. However, syntactically we have only that
$N(x)\succ\emptyset$ (that is: $N(x)$ is absolute), but not that 
$N(x)\succ\{x\}$. 
As a result, $\mathbb{N}$ is available in $RST^m$ 
(and its minimal model $J_1$) only as a proper $\safe$-class. 
In contrast,  $\mathbb{N}$ is available as a $\safe$-set 
in $RST_{HF}^m$,  
since it is definable  in all the universes (including of course
$J_2$) by the term:
$$\mathbb{\mbox{\ensuremath{\mathbb{\widetilde{N}}}}}:=\ \setin x{x\in HF\w N(x)} $$ 
Now we show that appropriate counterparts of Peano's axioms for $\mathbb{N}$ 
are provable in $RST_{HF}^m$. For this we need the following two lemmas.

\begin{lem}
\label{lem:HF-founded}
$\vdash_{RST_{HF}^m}\forall x\in H\!F\forall y. y\in x\imp  x\not\in y$.
\end{lem}

\begin{proof}
Let $\psi := \forall u\forall v. u\in z\w v\in u\imp u\not\in v$. 
We first show that  $\vdash_{RST_{HF}^m}\forall z\in H\!F\psi$.
By Lemma~\ref{lem:HF-1}(2), it suffices to show that 
$\ensuremath{\psi\left[\nicefrac{\emptyset}{z}\right]}$ and
$\forall x\forall y \left(\ensuremath{\psi\left[\nicefrac{x}{z}\right]}
\wedge\ensuremath{\psi\left[\nicefrac{y}{z}\right]}\imp
\ensuremath{\psi\left[\nicefrac{x\cup\;\setfin y}{z}\right]}
\right)$ are theorems of $RST_{HF}^m$. This is obvious for 
$\ensuremath{\psi\left[\nicefrac{\emptyset}{z}\right]}$.
To prove the other formula, we reason in $RST_{HF}^m$ as follows. Assume 
$\ensuremath{\psi\left[\nicefrac{x}{z}\right]}$ and
$\ensuremath{\psi\left[\nicefrac{y}{z}\right]}$. We show 
$\psi$ for $z=x\cup\;\setfin y$. So suppose
that $u\in z$ and $v\in u$. Then either $u\in x$ or $u=y$. In the first case
the assumptions $\ensuremath{\psi\left[\nicefrac{x}{z}\right]}$ and $v\in u$
implies that $u\not\in v$. In the second case $v\in y$, and so it follows
from the assumption $\ensuremath{\psi\left[\nicefrac{y}{z}\right]}$
that if $u\in v$ then $v\not\in u$, contradicting the assumption that $v\in u$.
Hence $u\not\in v$ in this case as well.

  To complete the proof of the lemma, let $x\in H\!F$. Then $\{x\}\in H\!F$
as well, and so $\ensuremath{\psi\left[\nicefrac{\{x\}}{z}\right]}$. Since
$x\in\{x\}$, this implies that $\forall y. y\in x\imp  x\not\in y$.
\end{proof}

\begin{lem}
\label{lem:HF-down}
$\vdash_{RST_{HF}^m}\forall x\in H\!F\forall y. y\subseteq x\imp  y\in H\!F$.
\end{lem}

\begin{proof}
Let $\psi:=\forall u. u\subseteq z\imp u\in H\!F$. 
We have to show that  $\vdash_{RST_{HF}^m}\forall z\in H\!F\psi$.
By Lemma~\ref{lem:HF-1}(2), it suffices to show that
$\ensuremath{\psi\left[\nicefrac{\emptyset}{z}\right]}$ and
$\forall x\forall y \left(\ensuremath{\psi\left[\nicefrac{x}{z}\right]}
\wedge\ensuremath{\psi\left[\nicefrac{y}{z}\right]}\imp
\ensuremath{\psi\left[\nicefrac{x\cup\;\setfin y}{z}\right]}
\right)$ are theorems of $RST_{HF}^m$. This is obvious for 
$\ensuremath{\psi\left[\nicefrac{\emptyset}{z}\right]}$.
To prove the other formula, we reason in $RST_{HF}^m$ as follows. Assume
$\ensuremath{\psi\left[\nicefrac{x}{z}\right]}$ and
$\ensuremath{\psi\left[\nicefrac{y}{z}\right]}$. We show
$\psi$ for $z=x\cup\;\setfin y$. So suppose
that $u\subseteq z$. Then there exists $v\subseteq x$ such that either
$u=v$ or $u=v\cup\;\setfin y$. Now the assumption that 
$\ensuremath{\psi\left[\nicefrac{x}{z}\right]}$ implies that
$v\in HF$, while the assumption that
$\ensuremath{\psi\left[\nicefrac{y}{z}\right]}$ implies that $y\in HF$.
Hence Lemma~\ref{lem:HF-1}(1) entails that $v\cup\;\setfin y\in HF$.
It follows that in both cases $u\in HF$. 
\end{proof}

\begin{prop}
\label{prop:NP}
\ 
\be
\item $\vdash_{RST_{HF}^m}0\in\mathbb{\mbox{\ensuremath{\mathbb{\widetilde{N}}}}}$
\item $\vdash_{RST_{HF}^m} \forall x.S(x)\neq 0$
\item $\vdash_{RST_{HF}^m} \forall x. 
x\in\mathbb{\mbox{\ensuremath{\mathbb{\widetilde{N}}}}}\leftrightarrow
S(x)\in\mathbb{\mbox{\ensuremath{\mathbb{\widetilde{N}}}}}$ 
\item $\vdash_{RST_{HF}^m} \forall x\in 
\mathbb{\mbox{\ensuremath{\mathbb{\widetilde{N}}}}} \forall y\in
\mathbb{\mbox{\ensuremath{\mathbb{\widetilde{N}}}}}.S(x)=S(y)\imp x=y$
\ee
\end{prop}

\begin{proof}
The first two items, and the fact that $RST_{HF}^m$ proves that
if $x\in\mathbb{\mbox{\ensuremath{\mathbb{\widetilde{N}}}}}$
then $S(x)\in\mathbb{\mbox{\ensuremath{\mathbb{\widetilde{N}}}}}$,
are very easy, and are left to the reader.

 To prove the other direction of (3), assume that 
$S(x)\in\mathbb{\mbox{\ensuremath{\mathbb{\widetilde{N}}}}}$. We show
that $x\in\mathbb{\mbox{\ensuremath{\mathbb{\widetilde{N}}}}}$.
Since $S(x)\in\mathbb{\mbox{\ensuremath{\mathbb{\widetilde{N}}}}}$,
by definition $S(x)\in HF$. Hence Lemma~\ref{lem:HF-down} implies 
that $x\in H\!F$  as well. To show
that also $N(x)$, let $y\in x\cup\{x\}$,
and suppose that $y\neq\emptyset$. These assumptions about $y$, and the assumption that
$S(x)\in\mathbb{\mbox{\ensuremath{\mathbb{\widetilde{N}}}}}$,  together imply that
there is $w\in S(x)$ such that $y=w\cup\{w\}$. It remains to show
that actually $w\in x$. This is obvious in case $y=x$, since $w\in y$
(because $y=w\cup\{w\}$). If $y\in x$ then $x\not\in y$ by Lemma~\ref{prop:NP} 
(because $x\in HF$), while the assumption that 
$y=w\cup\{w\}$ implies that $w\in y$. It follows that $w\neq x$.
Since $w\in S(x)$, this implies that $w\in x$ in this case too.

Finally, to prove item (4), we show that 
$\vdash_{RST_{HF}^m} \forall x\in HF \forall y.S(x)=S(y)\imp x=y$.
So suppose  that $x\in HF$ and $S(x)=S(y)$. Assume for contradiction that $x\neq y$.
Since $S(x)=S(y)$, this implies that both $x\in y$  and $y\in x$, which
is impossible by Lemma~\ref{prop:NP}.
\end{proof}


The induction rule is available in $RST_{HF}^m$ as well,
but only for $\varphi\safe\emptyset$. 
\begin{prop}
\label{prop:inductionN}$\vdash_{RST_{HF}^m}\left(\ensuremath{\varphi\left[\nicefrac{0}{x}\right]}
\wedge\forall 
x\left(\varphi\rightarrow\ensuremath{\varphi\left[\nicefrac{S\left(x\right)}{x}\right]}\right)
\right)\rightarrow\forall x\in\mathbb{\mbox{\ensuremath{\mathbb{\widetilde{N}}}}}.\varphi$,
for $\varphi\safe\emptyset.$
\end{prop}

\begin{proof}
Let $\varphi\safe\emptyset$, and assume that
$$ (*)\ \ \ensuremath{\varphi\left[\nicefrac{0}{x}\right]} \wedge\forall
x\left(\varphi\rightarrow\ensuremath{\varphi\left[\nicefrac{S\left(x\right)}{x}\right]}\right)$$
We show that $\forall x\in HF.\varphi$. 
Let $\psi$ be the formula 
$(x\in\mathbb{\mbox{\ensuremath{\mathbb{\widetilde{N}}}}}\imp\varphi)\w
\left(\forall z\in x.z\in\mathbb{\mbox{\ensuremath{\mathbb{\widetilde{N}}}}}
\imp\ensuremath{\varphi\left[\nicefrac{z}{x}\right]}\right)$.
Since $\varphi\safe\emptyset$, also $\psi\safe\emptyset$. Hence 
Lemma \ref{lem:HF-1}(2) implies: 
$$(**)\ \left(\ensuremath{\psi\left[\nicefrac{\emptyset}{x}\right]}
\wedge\forall x\forall y \left(\psi
\wedge\ensuremath{\psi\left[\nicefrac{y}{x}\right]}\imp
\ensuremath{\psi\left[\nicefrac{x\cup\;\setfin y}{x}\right]}
\right)\right)\rightarrow\forall x\in HF.\psi$$
Clearly $\psi\left[\emptyset\right]$
is provable in $RST_{HF}^m$,\footnote{To make the text more readable, at the
rest of the proof we write $\psi\left[t\right]$ and 
$\varphi\left[t\right]$ instead of 
$\ensuremath{\psi\left[\nicefrac{t}{x}\right]}$ and
$\ensuremath{\varphi\left[\nicefrac{t}{x}\right]}$ (respectively).}
 since we have $\varphi\left[\emptyset\right]$ by (*).
Now assume $\psi\left[x\right]\wedge\psi\left[y\right]$. We show
that $\psi\left[x\cup\setfin y\right]$. 
\be
\item Let $z\in x\cup\setfin y$, and suppose that 
$z\in\mathbb{\mbox{\ensuremath{\mathbb{\widetilde{N}}}}}$. Then 
either $z\in x$ or $z=y$, and in both cases the assumptions
$\psi\left[x\right]$, $\psi\left[y\right]$, and 
$z\in\mathbb{\mbox{\ensuremath{\mathbb{\widetilde{N}}}}}$ immediately imply
that $\varphi\left[z\right]$. 
\item Let $z= x\cup\setfin y$, and suppose that 
$z\in\mathbb{\mbox{\ensuremath{\mathbb{\widetilde{N}}}}}$.
The first of these assumptions implies that $z\neq\emptyset$, 
and so the second one entails that there exists $w$ such that
$z=w\cup\{w\}$. Then $w\in z$, and so 
$w\in\mathbb{\mbox{\ensuremath{\mathbb{\widetilde{N}}}}}\imp
\varphi\left[w\right]$ by item (1). But 
$w\in\mathbb{\mbox{\ensuremath{\mathbb{\widetilde{N}}}}}$ 
by Proposition~\ref{prop:NP}(2) and the assumption
$z\in\mathbb{\mbox{\ensuremath{\mathbb{\widetilde{N}}}}}$.
It follows that $\varphi\left[w\right]$. Therefore
(*) implies that $\varphi\left[z\right]$ in this case too. 
\ee
From (1) and (2) it follows that indeed 
$\psi\left[x\cup\setfin y\right]$ follows from
$\psi\left[x\right]\wedge\psi\left[y\right]$. Therefore (**) implies
that $\forall x\in HF.\psi$. This, in turn,
implies that $\forall x\in HF.\varphi$.
\end{proof}

\begin{rem}
\label{rem:res-ind}
The restriction to absolute formulas in Proposition~\ref{prop:inductionN}
is not a real problem for developing the theory of natural numbers 
that we need. With the help of Proposition~\ref{prop:+<} below,
Proposition~\ref{prop:inductionN} easily implies that all the formulas
in the language of first-order Peano's arithmetics and proofs in that theory 
can be translated into $RST_{HF}^m$ and its language.
This is because in this translation, all
the quantifications are bounded in $\mathbb{N}$, and thus they are
safe w.r.t. $\emptyset$. \footnote{It can be shown that the power
of full induction over $\mathbb{N}$ (i.e. for \emph{any} formula
$\varphi$) can be achieved by adding to $RST_{HF}^m$ the full
$\in$-induction scheme. } 
\end{rem}

Next we show that addition and multiplication on $\mathbb{N}$
are available in $RST_{HF}^{m}$ as small $\i$-functions.
In view of Propositions \ref{prop:NP} and \ref{prop:inductionN},
this suffices (as has been shown by G\"odel)
for having all recursive functions 
available in $RST_{HF}^{m}$ as small $\i$-functions.\footnote{Using
the method given in the proof of
Proposition~\ref{prop:+<} below, it is 
also not difficult to directly show 
that every primitive recursive function 
is available in $RST_{HF}^{m}$ as a small $\i$-function.}
\footnote{In 
\cite{Cohen2014QED} it was essentially shown (using a different
terminology) that every primitive recursive function is available
as a small $\i$-function in {\em the extension of $RST_{HF}^{m}$
with $\in$-induction}. Our present results show that $\in$-induction
is not really needed for this.}

\begin{prop}
\label{prop:+<}
\
\begin{enumerate}
\item The standard ordering $<$ on $\mathbb{N}$ is available in $RST_{HF}^m$
as a small $\i$-relation.
\item The standard addition and multiplication of natural numbers are available
in $RST_{HF}^m$ as small $\i$-functions.
\end{enumerate}
\end{prop}
\begin{proof}\leavevmode
\begin{enumerate}
\item The standard ordering $<$ on $\mathbb{N}$ coincides with $\in$.
Thus it is definable by the term 
$\setin{\left\langle m,n\right\rangle 
\check{\in}\widetilde{\mathbb{N}\times\mathbb{N}}}{ m\in n}$.
Since $\mathbb{N}$ is a $\i$-set, so is $\mathbb{N}\times\mathbb{N}$.
Hence the fact that $m\in n\safe\emptyset$ and Prop. \ref{Pro:SetPairs-1}
 imply that $<$ is a\emph{ }$\i$-set. It is now straightforward
to prove in $RST_{HF}^{m}$ its two characteristic  properties:
$\forall x.x\not<0$ and $\forall x\forall y.x<S(y)\leftrightarrow x<y\vee x=y$.
Using Proposition~\ref{prop:inductionN}, this suffices (as is well-known)
for deriving all the basic properties of <, like
its being a linear order or the existence of a <-successor for each element
in $\mathbb{N}$.
\item Define:
\bi
\item $Func\left(f\right):=\ 
\forall a,b,c\left(\left\langle a,b\right\rangle 
\check{\in}f\wedge\left\langle a,c\right\rangle 
\check{\in}f\rightarrow b=c\right)$
\item $add\left(z,u,n,f\right):=\
\left(z=0\wedge u=n\right)\vee\exists z_{1},u_{1}\in\widetilde{\mathbb{N}}\left(\left\langle z_{1},u_{1}\right\rangle \check{\in}f\wedge z=S\left(z_{1}\right)\wedge u=S\left(u_{1}\right)\right)$.
\item
$\psi_{+}(n,k,f):=
 Func\left(f\right)\wedge\forall x
\left(x\in f\leftrightarrow\exists z,u\in
\mbox{\ensuremath{\mathbb{\widetilde{N}}}}\left(z\leq k\wedge x=
\left\langle z,u\right\rangle \wedge add\left(z,u,n,f\right)\right)\right)$
\ei
Intuitively, $\psi_{+}(n,k.f)$ says that $f$ stands
for $\left\{ \left\langle 0,n\right\rangle ,\left\langle 1,n+1\right\rangle ,\left\langle 2,n+2\right\rangle ,...,\left\langle k,n+k\right\rangle \right\} $.
It is easy to check that $\psi_{+}(n,f)\safe\emptyset$. Hence
(Lemma \ref{lem:example}) addition is definable by:
\[
+:=\lambda n\in\mathbb{\mbox{\ensuremath{\mathbb{\widetilde{N}}}}},k\in\mathbb{\mbox{\ensuremath{\mathbb{\widetilde{N}}}}}.\iota m.\exists f\in HF\left(\psi_{+}(n,k,f)\wedge\left\langle k,m\right\rangle \check{\in}f\right).
\]
$+$ is a valid term since $f\in HF\safe\left\{ f\right\} $ and
$\left\langle k,m\right\rangle \check{\in}f\wedge
\psi_{+}(n,f)\safe\left\{ m\right\} $, and it 
is a small $\i$-function, as $\mathbb{N}\times\mathbb{N}$
is a $\i$-set. 

Multiplication is defined similarly. The only difference
 is that $add\left(z,u,n,f\right)$ is replaced with
$\left(z=0\wedge u=0\right)\vee\exists z_{1},u_{1}\in\mathbb{\mbox{\ensuremath{\mathbb{\widetilde{N}}}}}\left(\left\langle z_{1},u_{1}\right\rangle \check{\in}f\wedge z=S\left(z_{1}\right)\wedge u=n+u_{1}\right).$
This is legitimate, since $+$ is a small $\i$-function.

It is not difficult now to prove in $RST_{HF}^{m}$ 
the fundamental properties that characterize addition
and multiplication in first-order Peano's arithmetics:
\begin{eqnarray*}
\forall x.x+0=x & \,\,\,\mathrm{and}\,\,\, & \forall x,y.x+S\left(y\right)=S\left(x+y\right)\\
\forall x.x\cdot0=0 & \,\,\,\mathrm{and}\,\,\, & \forall x,y.x\cdot S\left(y\right)=x+\left(x\cdot y\right)
\end{eqnarray*}
Once these properties are proved, it is a standard matter to use
Prop. \ref{prop:inductionN} for proving all the standard properties of addition
and multiplication, such as commutativity, associativity, and distributivity. \qedhere
\ee
\end{proof}

\section{\label{sec:Real-Analysis}Real Analysis in \texorpdfstring{$J_{2}$}{J\_2}}

In this section we turn at last to the main goal of this paper: developing
real analysis in $J_{2}$. Now it
is not difficult to formalize the definitions, claims, and
proofs of this section in our formal framework.\emph{ }These translations
are straightforward, but rather tedious. Hence we shall
 omit them, with the exception of a few outlined examples.

\bs
\noindent \emph{Notation and Terminology}. Henceforth we restrict our attention to
the computational theory $RST_{HF}^m$ and its computational universe
$J_{2}$. Therefore 
we simply write $\left\Vert exp\right\Vert _{v}$
instead of $\left\Vert exp\right\Vert _{v}^{J_{2}}$. Similarly, 
when we talk about a 
$\safe$-set or a $\safe$-class, we mean $\safe$-set/$\safe$-class in $J_2$.


\subsection{The Construction of the Real Line}
 \ \ms

The standard construction of $\mathbb{Z}$, the set of integers, as
the set of ordered pairs $\left(\mathbb{N}\times\left\{ 0\right\} \right)\cup\left(\left\{ 0\right\} \times\mathbb{N}\right)$
can be easily carried out in $RST_{HF}^m$, as can the usual construction
of $\mathbb{Q}$, the set of rationals, in terms of ordered pairs
of relatively prime integers. There is also no difficulty in defining
the standard orderings on $\mathbb{Z}$ and $\mathbb{Q}$ as small
$\i$-relations, as well as the standard functions of addition and
multiplication as small $\i$-functions. The main properties of addition and multiplication are provable in
$RST_{HF}^m$, as the standard proofs by induction can be carried
out within it. Furthermore, all the basic properties of $\mathbb{Z}$
and $\mathbb{Q}$ (such as $\mathbb{Q}$ being a dense unbounded field)
are straightforwardly proved in $RST_{HF}^m$. 

Now we turn to the standard construction of the real line using Dedekind
cuts. Since it is well known that the real line and its open segments
are not absolute, they cannot be $\safe$-sets. 
Thus the collection of real numbers in $RST_{HF}^m$ will not
be a term but merely a \emph{definable predicate.}, that is: a 
{\em $\safe$-class}.\footnote{As noted in 
Footnote \ref{fn:counterpart}, this is in sharp difference from
the development of real analysis in \cite{Cohen2014QED}.}

\bs

Let $\psi:=\forall x,y\in\widetilde{\mathbb{Q}}.x\in u\wedge y<x\imp y\in u$, 
$\ \ \ \varphi:=\neg\exists x\in u\forall y\in u.y\leq x$.
\begin{defi}
[The Reals]\label{Def Reals} $\mathbb{R}$ is 
$\left\Vert \class{u\in\overline{P_{J_{2}}\left(\mathbb{Q}\right)\setminus
\left\{ \emptyset,\mathbb{Q}\right\} }}{\psi\wedge\varphi}\right\Vert $.
\end{defi} 
The above term is a valid class term as $P_{J_{2}}\left(\mathbb{Q}\right)\setminus\left\{ \emptyset,\mathbb{Q}\right\} $
is a $\safe$-class, and $\varphi,\psi\safe\emptyset$.
Note that it does not denote 
the ``real'' real-line (if such a thing really exists). However,
it does contain all \emph{computable} real numbers, such as $\sqrt{2}$
and $\pi$. (This can be shown by the same method that was used in  \cite{Cohen2014QED}.)

\ms

\noindent \emph{Notation}. We employ the following notations: $\mathbb{Q}^{+}=\left\{ q\in\mathbb{Q}\mid0<q\right\} $,
$\mathbb{R}^{+}=\left\{ r\in\mathbb{R}\mid0<r\right\} $, $\left(a,b\right)=\left\{ r\in\mathbb{R}\mid a<r<b\right\} $
and $\left[a,b\right]=\left\{ r\in\mathbb{R}\mid a\leq r\leq b\right\} $,
for $a,b$ real numbers.
\footnote{Notice that $\mathbb{Q}^{+}$ is a $\i$-set and $\mathbb{R}^{+}$ is a $\i$-class.}
\begin{prop}
\label{prop:R operations}The following holds: 
\begin{enumerate}
\item The standard ordering $<$ on $\mathbb{R}$ is available in $RST_{HF}^m$
as a \emph{$\i$}-relation. 
\item The standard addition and multiplication of reals are available in
$RST_{HF}^m$ as $\i$-functions. 
\end{enumerate}
\end{prop}
\begin{proof}\leavevmode
\begin{enumerate}
\item The relation $<$ on $\mathbb{R}$ coincides with $\subset$, thus
we can define the relation $<$ by $\class{\left\langle x,y\right\rangle \in\overline{\mathbb{R}\times\mathbb{R}}}{x\subset y}$.
We have that $x\subset y\safe\emptyset$, hence $<$ is a \emph{$\safe$}-class.
It is straightforward to prove in $RST_{HF}^m$ properties concerning
$<$, such as it being a total order on $\mathbb{R}$, the density
of the rationals in $\mathbb{R}$, the Archimedean Principle, etc.
\item The $\safe$-function $+$ can be represented (using Lemma \ref{lem:Functions Severl variables})
by the term 
\[
+=\lambda x\in\mathbb{\bar{R}},y\in\bar{\mathbb{R}}.\setin z{\exists u\in x\exists v\in y.z=u+v}
\]
since $\exists u\in x\exists v\in y.z=u+v\safe\left\{ z\right\} $.
\smallskip

To define multiplication, let $F_{1}$ be the \emph{$\safe$}-function:
\[F_{1}=  \left\Vert \lambda a\in\bar{\mathbb{R}^{+}},b\in\bar{\mathbb{R}^{+}}.
\setin z{z\leq 0\vee\exists u\in a\exists v\in b\left(0\leq u\wedge0\leq v\wedge 
z=u\cdot v\right)}\right\Vert\]
Next, define the \emph{$\safe$}-function $-$ on $\mathbb{R}$ by
\[
-=\left\Vert \lambda x\in\bar{\mathbb{R}}.\class z{\exists u\in\widetilde{\mathbb{Q\setminus}x}\exists a\in\widetilde{\mathbb{Q}}.z+b=a}\right\Vert .
\]
Then, for $0\leq a\wedge b<0$ define $F_{2}\left(\left\langle a,b\right\rangle \right):=-F_{1}\left(\left\langle a,-b\right\rangle \right)$,
for $a<0\wedge0\leq b$ define $F_{3}\left(\left\langle a,b\right\rangle \right):=-F_{1}\left(\left\langle -a,b\right\rangle \right)$,
and for $a<0\wedge b<0$ define $F_{4}\left(\left\langle a,b\right\rangle \right):=F_{1}\left(\left\langle -a,-b\right\rangle \right)$.
Now the $\safe$-function $\cdot$ on $\mathbb{R}\times\mathbb{R}$
can be defined by $\cdot:=F_{1}\cup F_{2}\cup F_{3}\cup F_{4}$. \\
Proving in $RST_{HF}^m$ basic properties regarding these $\i$-functions,
such as $\mathbb{R}$ being an ordered field, is again straightforward.
\qedhere
\end{enumerate}
\end{proof}

\subsection{The  Least Upper Bound Principle and the Topology of the Reals}
\ \ms

In this section  we examine to what extent the least upper bound principle is available in
$RST_{HF}^m$. We start with the following positive result: 
\begin{thm}
\label{prop:R-completeness}It is provable in $RST_{HF}^m$ that
every nonempty $\i$-subset of $\mathbb{R}$ that is bounded above
has a least upper bound in $\mathbb{R}$. Furthermore, the induced mapping  ($l.u.b$)
 is available in $RST_{HF}^m$ as a $\i$-function.
\end{thm}

\begin{proof}
Let $X$ be a nonempty $\i$-subset of $\mathbb{R}$ that is bounded
above. $\cup X$ is a $\i$-set, and since $X$ is bounded above,
standard arguments show that $\cup X$ is a Dedekind cut and thus
belongs to $\mathbb{R}$. Since the order $\i$-relation $\leq$ coincides
with the inclusion relation, it follows that $\cup X$ is a least
upper bound for $X$. Moreover, the function that maps each $X$ to
$\cup X$ is a rudimentary function (from $P_{J_{2}}\left(\mathbb{R}\right)$
to $P_{J_{2}}\left(\mathbb{Q}\right)$), and hence it is a $\i$-function.
Denote it by $F$. The desired function $l.u.b$ is $F\upharpoonright_{F^{-1}\left[\mathbb{R}\right]}$,
which by Proposition \ref{prop:functions} is a $\i$-function.
\end{proof}

Encouraging as it may be, Theorem \ref{prop:R-completeness} only
states that \emph{$\safe$-subsets} of $\mathbb{R}$ have the least
upper bound property. Therefore it is insufficient for 
developing in $RST_{HF}^m$ most of standard mathematics. The reason is
that even the most basic substructures of $\mathbb{R}$, like the
intervals, are not $\safe$-sets in $RST_{HF}^m$, but only proper $\i$-classes.
Hence, a stronger version of the theorem, which ensures that
the least upper bound property holds for standard $\safe$-subclasses
of $\mathbb{R}$, is needed. Theorem \ref{thm:lubinterval}
below provides such an extension, but it requires some additional definitions and
propositions.\footnote{It should be noted that the full least upper bound principle has 
not been derivable also in Weyl's approach \cite{weyl1932kontinuum}.
To obtain the principle for standard mathematical objects, we use in what follows 
coding techniques that are similar to those employed by Weyl.}

\ms

First we consider $\i$-classes $U\subseteq\mathbb{R}$ which are
open\emph{.} These $\i$-classes are generally not $\i$-sets (unless
empty), since they contain an interval of positive length, which is
a proper $\i$-class and thus cannot be contained in a $\i$-set (see
Prop. \ref{prop:Basic setclass}(3)). Clearly, there is no such thing
as a $\i$-set of $\i$-classes, as a proper $\i$-class can never
be an element of another $\i$-set or $\i$-class. However, the use
of coding (following \cite{simpson2009subsystems}, \cite{weaver2005analysis}\footnote{In \cite{weaver2005analysis} such codings are called ``proxies''.})
allows us, for example, to replace the meaningless statement ``the
union of a $\i$-set of $\i$-classes is a $\i$-class'' with ``given
a $\i$-set of codes for $\i$-classes, the union of the corresponding
$\i$-classes is a $\i$-class''.

The coding technique we use is based on  the standard mathematical 
notation  for a ``family of sets'', $\left(A_{i}\right)_{i\in I}$,
where $I$ is a set of indices and $A_{i}$ is a set for each $i\in I$.
In $RST_{HF}^m$ we cannot construct the collection of all such
$A_{i}$'s if $A_{i}$ is a $\i$-class for some $i\in I$. Thus,
we treat the $\i$-set $I$ as a code for the ``family of classes''
$\left(A_{i}\right)_{i\in I}$. In fact, we mainly use the union of
such families, i.e., $\underset{i\in I}{\bigcup A_{i}}$.

\begin{defi}
For $p,q\in\mathbb{R}$, 
the open ball $B_{q}\left(p\right)$
is the $\i$-class $\left\{ r\in\mathbb{R}\mid\left|r-p\right|<q\right\} $.
\end{defi}
\begin{defi}
Let $U\subseteq\mathbb{R}$ be a $\i$-class. If there exists a $\i$-set
$u\subseteq\mathbb{Q}\times\mathbb{Q}^{+}$ such that  
$U=\bigcup_{\left\langle p,q\right\rangle \in u}B_{q}\left(p\right)
=\left\{ r\in\mathbb{R}\mid\exists p,q\left(\left\langle p,q\right\rangle 
\in u\wedge\left|r-p\right|<q\right)\right\}$,
then $U$ is called \emph{open} and $u$ is a \emph{code} for $U$.
\end{defi}
In what follows, the formalizations in $RST_{HF}^m$ are carried
out as follows: 
\begin{itemize}
\item To quantify over open $\i$-classes: $Qu\subseteq\widetilde{\mathbb{Q}\times\mathbb{Q}^{+}}$
($Q\in\left\{ \forall,\exists\right\} $). 
\item To decode the open $\i$-class whose code is $u$:  
\[
dec\left(u\right):=\class{r\in\bar{\mathbb{R}}}{\exists p,q\left(\left\langle p,q\right\rangle \check{\in}u\wedge\left|r-p\right|<q\right)}
\]  
\item To state that a class variable $\boldsymbol{U}$ is an open $\i$-class:
\[
Open\left(\boldsymbol{U}\right):=\exists u\subseteq\widetilde{\mathbb{Q}\times\mathbb{Q}^{+}}.\boldsymbol{U}=dec\left(u\right)
\]
\end{itemize} 
\begin{prop}
\label{prop:open}The following are provable in $RST_{HF}^m$: 
\begin{enumerate}
\item For any  $\i$-set $u\subseteq\mathbb{R}\times\mathbb{R}^{+}$, $\left\{ r\in\mathbb{R}\mid\exists p,q\left(\left\langle p,q\right\rangle \in u\wedge\left|r-p\right|<q\right)\right\} $
is an open $\i$-class.
\item The open ball $B_{q}\left(p\right)$ is an open $\i$-class for any
$p\in\mathbb{R}$ and $q\in\mathbb{R}^{+}$.
\end{enumerate}
\end{prop}
\begin{proof}\leavevmode
\begin{enumerate}
\item Take $w$ to be $\left\Vert \setin{\left\langle p,q\right\rangle \in\widetilde{\mathbb{Q}\times\mathbb{Q}^{+}}}{\exists r,s\left(\left\langle r,s\right\rangle \check{\in}u\wedge q+\left|r-p\right|\leq s\right)}\right\Vert $.
Since $\mathbb{Q}\times\mathbb{Q}^{+}$ is a $\i$-set and $\exists r,s\left(\left\langle r,s\right\rangle \check{\in}u\wedge q+\left|r-p\right|\leq s\right)\safe\emptyset$,
$w$ is a $\i$-set that is a code for an open $\i$-class. It can
easily be proved in $RST_{HF}^m$ that $w$ codes $U$. 
\item $u=\left\{ \left\langle p,q\right\rangle \right\} $ is a code of
$B_{q}\left(p\right)$ (by 1.).\qedhere
\end{enumerate}
\end{proof}

\begin{prop}
\label{prop:class-union}
The following are provable in $RST_{HF}^m$:\label{prop:proxy-1} 
\begin{enumerate}
\item The union of a $\i$-set of open $\i$-classes is an open $\i$-class.
i.e, given a $\i$-set of codes of open $\i$-classes, the union of
the corresponding open $\i$-classes is an open $\i$-class. 
\item The intersection of finitely many  open $\i$-classes is an open $\i$-class. 
\end{enumerate}
\end{prop}
\begin{proof}
\begin{enumerate}
\item Let $X$ be a $\i$-set of codes for open $\i$-classes. Thus, $\cup X$
is a code for the union of the corresponding open $\i$-classes. 
\item If $U$ and $V$ are open $\i$-classes, a code for their intersection
is obtained by intersecting every ball in a code for $U$ with every
ball in a code for $V$.\qedhere
\end{enumerate}
\end{proof}

\begin{rem}
\label{Note:scheme}\emph{In general,} when we say that a theorem about a $\i$-class
or a $\i$-function is provable in $RST_{HF}^m$ (as in Prop. \ref{prop:class-union})
we mean that it can be formalized and proved as a
{\em scheme}, that is: its proof can be carried out in $RST_{HF}^m$
using a uniform scheme. However, propositions about open
$\i$-classes form an exception, because due to the coding machinery, they can be fully formalized
and proved in $RST_{HF}^m$.
\end{rem}

\begin{exa}
\label{example1}As an example of the use of the coding technique,
we demonstrate the formalization of Prop. \ref{prop:proxy-1}(1): 
\begin{gather*}
\forall z.(\forall x\in z.x\subseteq\widetilde{\mathbb{Q}\times\mathbb{Q}^{+}})\rightarrow\exists w\subseteq\widetilde{\mathbb{Q}\times\mathbb{Q}^{+}}.dec\left(w\right)=\class r{\exists x\in z.r\in dec\left(x\right)}
\end{gather*}
\end{exa}
\begin{defi}
A $\i$-class $X\subseteq\mathbb{R}$ is \emph{closed} if $\mathbb{R}-X$
is open.
\end{defi}
\begin{lem}
\label{lem:open ball about}It is provable in $RST_{HF}^{m}$ that
if $U\subseteq\mathbb{R}$ is an open $\i$-class, then for every
$x\in U$ there is an open ball about $x$ which is contained in $U$.
\end{lem}
\begin{proof}
If $x\in U$, then there is some $\left\langle p,q\right\rangle $
in the code of $U$ such that $x\in B_{q}\left(p\right)$. Take $\varepsilon=\left|p-x\right|$.
It is straightforward to see that 
$B_{\varepsilon}\left(x\right)\subseteq B_{q}\left(p\right)\subseteq U$.
\end{proof}

The proof of the next Lemma is trivial.

\begin{lem}
\label{lem:dense}Let $X\subseteq\mathbb{R}$ be a $\i$-class and
$A\subseteq X$ be a $\i$-set. The following are equivalent in $RST_{HF}^m$: 
\begin{enumerate}
\item Every open ball about a point in $X$ intersects $A$.
\item Every open $\i$-class that intersects $X$ also intersects $A$.
\end{enumerate}
\end{lem}
\begin{exa}
\noindent \label{Example2}As an example of a full formalization which
uses class variables, the formalization of the Lemma above is: 
{\small{}
\begin{gather*}
\phi:=\boldsymbol{X}\subseteq\bar{\mathbb{R}}\rightarrow\forall a\subseteq\boldsymbol{X}\left(\forall x\in\boldsymbol{X}\forall\varepsilon\in\bar{\mathbb{R}^{+}}\left(B_{\varepsilon}\left(x\right)\cap a\neq\emptyset\right)\leftrightarrow\right.\\
\left.\forall u\subseteq\widetilde{\mathbb{Q}\times\mathbb{Q}^{+}}\left(dec\left(u\right)\cap\boldsymbol{X}\neq\emptyset\rightarrow dec\left(u\right)\cap a\neq\emptyset\right)\right)
\end{gather*}
}
We now demonstrate how to obtain an equivalent schema in the basic $\lng$
by replacing each appearance of a class term or a variable with the
formula it stands for. First, we explain the translation
of $x\in\bar{\mathbb{R}}$ to $\lng$ . One iteration of the translation
entails $x\in\overline{P_{J_{2}}\left(\mathbb{Q}\right)\setminus
\left\{ \emptyset,\mathbb{Q}\right\} }\wedge
\varphi\left(x\right)\w
\psi\left(x\right)$
for $\varphi,\psi$  as in Def. \ref{Def Reals}. A second iteration
yields $R\left(x\right):=x\subseteq\tilde{\mathbb{Q}}\wedge x\neq\tilde{\mathbb{Q}}\wedge 
x\neq\emptyset\wedge \varphi\left(x\right)\wedge\psi\left(x\right)$
which is in $\lng$. For the  translation of $\phi$,
first substitute $\class x{\theta}$ for $\boldsymbol{X}$, where $\theta\safe\emptyset$.
Proceeding with the translation steps results in the following formula
(scheme) of $\lng$, for $\theta\safe\emptyset$: {\small{}
\begin{gather*}
{\normalcolor \forall b\left(\theta\left(b\right)\rightarrow R\left(b\right)\right)\rightarrow\forall a\left(\left(\forall z.z\in a\rightarrow\theta\left(z\right)\right)\rightarrow\forall x\left(\theta\left(x\right)\rightarrow\forall\varepsilon\left(\left(R\left(\varepsilon\right)\wedge0<\varepsilon\right)\rightarrow\right.\right.\right.}\\
{\normalcolor \exists w.\left|w-x\right|<\varepsilon\wedge w\in a\leftrightarrow\forall u\subseteq\widetilde{\mathbb{Q}\times\mathbb{Q}^{+}}\left(\exists w.R\left(w\right)\wedge\exists p,q\left(\left\langle p,q\right\rangle \check{\in}u\wedge\left|w-p\right|<q\right)\wedge\right.}\\
{\normalcolor \left.\left.\theta\left(w\right)\right)\rightarrow\exists w.R\left(w\right)\wedge\exists p,q\left(\left\langle p,q\right\rangle \check{\in}u\wedge\left|w-p\right|<q\right)\wedge w\in a\right)}
\end{gather*}
}{\small \par}
\end{exa}

\begin{defi}
Let $X\subseteq\mathbb{R}$ be a $\safe$-class, and $A\subseteq X$
a $\safe$-set. $A$ is called \emph{dense in $X$ }if one of the
conditions of Lemma \ref{lem:dense} holds. $X$ is called \emph{separable}
if it contains a dense $\safe$-subset.
\end{defi}
\begin{prop}
\label{prop: subsepar}It is provable in $RST_{HF}^m$ that an
open $\safe$-subclass of a separable $\safe$-class is separable.
\end{prop}

\begin{proof}
Let $X$ be an open $\safe$-subclass of the separable $\safe$-class
$S$, and let $D$ be the dense $\safe$-subset in $S$. Let $B$
be an open ball with center $x\in X$. By Lemma~\ref{lem:open ball about} there
is a ball $B'$ about $x$ s.t. $B'\subseteq B\cap X$. Since $D$
is dense in $S$, $B'\cap D\neq\emptyset$. Hence, $B\cap D\cap X\neq\emptyset$,
and so $D\cap X$ is dense in $X$. 
\end{proof}

\noindent Now we can finally turn to prove a more encompassing least upper bound theorem.
\begin{thm}
\label{thm:lubinterval}It is provable in $RST_{HF}^m$ that every
nonempty separable $\safe$-subclass of $\mathbb{R}$ that is bounded
above has a least upper bound in $\mathbb{R}$.
\end{thm}

\begin{proof}
Let $X$ be a nonempty separable $\safe$-subclass of $\mathbb{R}$
that is bounded above, and let $A$ be a dense $\safe$-subset in
$X$. By Theorem \ref{prop:R-completeness}, $A$ has a least upper
bound, denote it $m$. Suppose $m$ is not an upper bound of $X$,
i.e. there exists $x\in X-A$ s.t. $x>m$. Take $\varepsilon_{m}=\left|x-m\right|$.
Then, there is no $a\in A$ s.t. $a\in A\cap\left(x-\varepsilon_{m},x+\varepsilon_{m}\right)$,
which contradicts $A$ being dense in $X$. That $m$ is the least
upper bound of $X$ is immediate.
\end{proof}

\begin{exa}
\label{Example3}To demonstrate the formalization in $\lng$ of the
last theorem, denote by $separ\left(\boldsymbol{U}\right)$ the formula
$\exists d.d\subseteq\boldsymbol{U}\wedge\forall x\in\boldsymbol{U}\forall\varepsilon\in\bar{\mathbb{R}^{+}}B_{\varepsilon}\left(x\right)\cap\boldsymbol{U}\neq\emptyset$,
and by $bound_{\boldsymbol{U}}\left(w\right)$ the formula $\forall x\in\boldsymbol{U}.x\leq w$.
Now, the full formalization is:
\begin{gather*}
\left(\boldsymbol{U}\subseteq\mathbb{R}\wedge\boldsymbol{U}\neq\emptyset\wedge separ\left(\boldsymbol{U}\right)\wedge\exists w\in\bar{\mathbb{R}}.bound{}_{\boldsymbol{U}}\left(w\right)\right)\rightarrow\\
\exists v\in\bar{\mathbb{R}}\left(bound_{\boldsymbol{U}}\left(v\right)\wedge\forall w\in\bar{\mathbb{R}}\left(bound_{\boldsymbol{U}}\left(w\right)\rightarrow v\leq w\right)\right)
\end{gather*}
\end{exa}

\begin{defi}
\label{Def: interval}A $\safe$-class $X\subseteq\mathbb{R}$ is
called an \emph{interval }if\emph{ }for any $a,b\in X$ s.t. $a<b$:
if $c\in\mathbb{R}\wedge a<c<b$ then $c\in X$.
\end{defi}


\begin{prop}
\label{prop:separ}It is provable in $RST_{HF}^{m}$ that any non-degenerate
interval is separable.
\end{prop}
\begin{proof}
Let $X$ be a non-degenerate interval. Take $A$ to be $X\cap\mathbb{Q}$.
By Prop. \ref{prop:Basic setclass}(2) $A$ is a $\safe$-set. A standard
argument shows that in every open ball about a point in $X$ there
is a rational number, and thus it intersects $A$.
\end{proof}
\begin{cor}
\label{cor:interval lub}It is provable in $RST_{HF}^{m}$ that
any non-degenerate interval that is bounded above has a least upper
bound.
\end{cor}

\begin{defi}
A $\safe$-class $X\subseteq\mathbb{R}$ is called \emph{connected
}if there are no open $\safe$-classes $U$ and $V$ such that $X\subseteq U\cup V$,
$U\cap V\neq\emptyset$, $X\cap U\neq\emptyset$ and $X\cap V\neq\emptyset$.
\end{defi}
\begin{exa}
The formalization of the above definition can be given by:
\begin{gather*}
connected\left(\boldsymbol{X}\right):=\neg\exists u,v\subseteq\widetilde{\mathbb{Q}\times\mathbb{Q}^{+}}\left(\boldsymbol{X}\subseteq decode\left(u\right)\cup decode\left(v\right)\wedge\right.\\
\left.decode\left(u\right)\cap decode\left(v\right)\neq\emptyset\wedge\boldsymbol{X}\cap decode\left(u\right)\neq\emptyset\wedge\boldsymbol{X}\cap decode\left(v\right)\neq\emptyset\right)
\end{gather*}
\end{exa}
\begin{prop}
\label{prop:interval connected}Let $X\subseteq\mathbb{R}$ be a $\safe$-class.
It is provable in $RST_{HF}^{m}$ that $X$ is connected if and
only if it is an interval.
\end{prop}
\begin{proof}
Assume that there are open $\safe$-classes $U$ and $V$ such that
$X\subseteq U\cup V$, $U\cap V\neq\emptyset$, $X\cap U\neq\emptyset$,
and $X\cap V\neq\emptyset$ (recall that the formalization of the
existence of open $\safe$-classes is done using their codes). Choose
$u\in U$ and $v\in V$ and assume that $u<v$. Let $U_{0}=U\cap\setout{z\in\mathbb{R}}{z<v}$
and $V_{0}=V\cap\setout{z\in\mathbb{R}}{z>u}$. Prop. \ref{prop:proxy-1}
and Prop. \ref{prop: subsepar} entail that $U_{0}$ and $V_{0}$
are open, separable $\safe$-subclass of $\mathbb{R}$. Standard arguments
show that they are non-empty and bounded above. Thus, by Theorem \ref{thm:lubinterval},
$U_{0}$ and $V_{0}$ have least upper bounds. Following the standard
proof found in ordinary textbooks we can deduce that the least upper
bounds are elements in $\left[u,v\right]$, but not elements of $U_{0}$
or of $V_{0}$, which is a contradiction, since $\left[u,v\right]\subseteq U_{0}\cup V_{0}$.
 The classical proof of the converse direction can easily be carried
out in $RST_{HF}^{m}$.
\end{proof}

\subsection{Real Functions}
\begin{defi}
\label{seq}Let $X$ be a $\i$-class. A $\i$-\emph{sequence} in
$X$ is a $\i$-function on $\mathbb{N}$ whose image is contained
in $X$.
\end{defi}
\begin{lem}
\label{lem:Cauchy}It is provable in $RST_{HF}^m$ that 
Cauchy $\i$-sequences in $\mathbb{R}$ converge to limits in $\mathbb{R}$.
The induced map ($lim$) 
 is available in $RST_{HF}^m$ as a $\i$-function.
\end{lem}

\begin{proof}
Let $a$ be a Cauchy $\i$-sequence, and let $a_{k}$ abbreviate $a\left(k\right)$.
For $n\in\mathbb{N}$ define $v_{n}:=\bigcap_{k\geq n}a_{k}$. The
$l.u.b$ of $\lambda n.v_{n}$ equals the limit of $\lambda n.a_{n}$.
Thus, $lim\,\lambda n.a_{n}:=\underset{}{\bigcup}\setin{v_{n}}{n\in\widetilde{\mathbb{N}}}.$
\end{proof}

\begin{prop}
\label{prop:closeCauchy-1-1}It is provable in $RST_{HF}^m$ that
if $X\subseteq\mathbb{R}$ is closed, then every Cauchy $\i$-sequence
in $X$ converges to a limit in $X$.
\end{prop}

\begin{proof}
Let $a$ be a Cauchy $\i$-sequence in $X$, and let $a_{k}$ abbreviate
$a\left(k\right)$. By Lemma \ref{lem:Cauchy}, $lim\,\lambda n.a_{n}$
is an element in $\mathbb{R}$, denote it by $l$. Assume by contradiction
that $l\in\mathbb{R}-X$. Since $X$ is closed, $\mathbb{R}-X$ is
open, and thus there exists $\varepsilon>0$ such that $B_{\varepsilon}\left(l\right)\subseteq\mathbb{R}-X$.
From this follows that for every $a_{k}$, $a_{k}\notin B_{\varepsilon}\left(l\right)$,
which contradicts the fact that $lim\,\lambda n.a_{n}=l$.
\end{proof}

Next we want to study sequences of functions, but
 Def. \ref{seq} cannot be applied as is, since $\i$-functions
which are proper $\i$-classes cannot be values of a $\i$-function (in particular,
of a \emph{$\i$-}sequence). Instead, we use the standard
Un-currying procedure.
\begin{defi}
For $X,Y$  $\i$-classes, a $\i$-\emph{sequence of }$\i$-\emph{functions}
on $X$ whose image is contained in $Y$ is a $\i$-function on $\mathbb{N}\times X$
with image contained in $Y$. 
\end{defi}
\begin{prop}
\label{prop:Pointwise}
Any point-wise limit of a $\i$-sequence of $\i$-functions on a $\i$-class $X\subseteq\mathbb{R}$
whose image is contained in $\mathbb{R}$ is available in $RST_{HF}^m$
as a $\i$-function.
\end{prop}

\begin{proof}
Let $F$ be a $\i$-sequence of $\i$-functions on $X$ whose image
is contained in $\mathbb{R}$. Suppose that for each $a\in X$ the
$\i$-sequence $\lambda n.F\left(n,a\right)$ is converging, and so
it is Cauchy. Define: $G_{a}:=\setin{\left\langle n,\bar{F}\left(n,a\right)\right\rangle }{n\in\widetilde{\mathbb{N}}}$.
Then, $\left\Vert \lambda a\in\bar{X}.limG_{a}\right\Vert $ is the
desired $\i$-function.
\end{proof}

Next we turn to continuous real $\safe$-functions. One possibility
of doing so, adopted e.g., in \cite{simpson2009subsystems,weyl1932kontinuum},
is to introduce codes for continuous real $\i$-functions (similar
to the use of codes for open $\i$-classes). This is of course possible
as such $\i$-functions are determined by their values on the $\i$-set
$\mathbb{Q}$. However, we prefer to present here another approach,
which allows for almost direct translations of proofs in standard
analysis textbook into our system. This is done using free function
variables. Accordingly, the theorems which follow are schemes. Implicitly,
the previous sections of this paper can also be read and understood
as done in this manner. Therefore, in what follows we freely use results
from them.
\begin{defi}
\label{Def cont}Let $X\subseteq\mathbb{R}$ be a $\safe$-class and
let $F$ be a $\i$-function on $X$ whose image is contained in $\mathbb{R}$.
$F$ is called a \emph{continuous real $\i$-function }if: 
\[
\forall a\in X\forall\varepsilon\in\mathbb{R}^{+}\exists\delta\in\mathbb{R}^{+}\forall x\in X.\left|x-a\right|<\delta\rightarrow\left|F\left(x\right)-F\left(a\right)\right|<\varepsilon
\]
\end{defi}
\begin{prop}
\label{prop:epsilon}Let $X\subseteq\mathbb{R}$ be a $\safe$-class
and $F$ be a $\i$-function on $X$ whose image is contained in $\mathbb{R}$.
It is provable in $RST_{HF}^m$ that if for every open $\safe$-class
$B\subseteq\mathbb{R}$, there is an open $\i$-class $A$ s.t. $F^{-1}\left[B\right]=A\cap X$,
then $F$ is continuous.
\end{prop}

\begin{proof}
Let $a\in X$, $\varepsilon>0$, and $V=B_{\varepsilon}\left(F\left(a\right)\right)$.
Since $V$ is an open $\i$-class, there is an open $\i$-class $A$
s.t. $F^{-1}\left[V\right]=A\cap X$ (which is a $\i$-class by Prop.
\ref{prop:functions}(3)). Also, $F\left(a\right)\in V$ which entails
$a\in F^{-1}\left[V\right]$, and thus $a\in A$. Since $A$ is open
there exists $\delta_{a}$ s.t. $B_{\delta_{a}}\left(a\right)\subseteq A$.
Take $\delta=\delta_{a}$. For any $x\in X$, if $\left|x-a\right|<\delta_{a}$
then $x\in B_{\delta_{a}}\left(a\right)\subseteq A$. Hence $x\in A\cap X=F^{-1}\left[V\right]$,
and therefore $F\left(x\right)\in V=B_{\varepsilon}\left(F\left(a\right)\right)$,
i.e. $\left|F\left(x\right)-F\left(a\right)\right|<\varepsilon$.
\end{proof}

\begin{lem}
\label{prop:comp sum prod}The following are provable in $RST_{HF}^m$: 
\begin{enumerate}
\item The composition, sum and product of two continuous real $\safe$-functions
is a continuous real $\safe$-function. 
\item The uniform limit of a $\i$-sequence of continuous real $\safe$-functions
is a continuous real $\safe$-function. 
\end{enumerate}
\end{lem}
\begin{proof}
The standard proofs of these claims can be easily carried out in
$RST_{HF}^{m}$. Note that they require the triangle inequality
which is provable in $RST_{HF}^{m}$.
\end{proof}

Next we prove, as examples, the Intermediate Value Theorem and the
Extreme Value Theorem, which are two key properties of continuous
real functions.

\begin{thm}
[Intermediate Value Theorem]\label{thm:Intermediate} Let $F$ be
a continuous real $\i$-function on an interval $\left[a,b\right]$ with $F\left(a\right)<F\left(b\right)$. It is provable in
$RST_{HF}^m$ that for any $d\in\mathbb{R}$ s.t. $F\left(a\right)<d<F\left(b\right)$,
there is $c\in\left[a,b\right]$ s.t. $F\left(c\right)=d$.
\end{thm}

\begin{proof}
Let $d\in\mathbb{R}$ such that $F\left(a\right)<d<F\left(b\right)$.
Define
\begin{displaymath}
  Q_{d}:=\left\Vert \setin{x\in\mathbb{\tilde{Q}}}{x\in\overline{\left[a,b\right]}\wedge F\left(x\right)\le d}\right\Vert.
\end{displaymath}
$Q_{d}$ is clearly bounded (e.g. by $b$). Since $F\left(a\right)<d$
, standard arguments that use the continuity of $F$ and the denseness
of $\mathbb{Q}$ in $\mathbb{R}$ show that there is a rational $a\leq q$
s.t. $F\left(q\right)\leq d$. Thus, $Q_{d}$ is non-empty and by
Thm. \ref{prop:R-completeness} it has a least upper bound, denote
it by $c$. Since $Q_{d}$ is non-empty and $b$ is an upper bound
for it, $c\in\left[a,b\right]$. Assume by contradiction that $F\left(c\right)<d$,
and pick $\varepsilon=d-F\left(c\right)$. By the continuity of $F$
there exists $\delta>0$ s.t. for any $x\in\left[a,b\right]$, if
$\left|x-c\right|<\delta$, then $\left|F\left(x\right)-F\left(c\right)\right|<\varepsilon=d-F\left(c\right)$.
This yields the existence of a rational $q\in\left(c,c+\delta\right)$
(again, by the denseness of $\mathbb{Q}$ in $\mathbb{R}$) s.t. $F\left(q\right)<d$,
which is a contradiction. Now, assume by contradiction that $F\left(c\right)>d$,
and pick $\varepsilon=F\left(c\right)-d$. In this case there exists
$\delta>0$ s.t. for any $x\in\left[a,b\right]$, if $\left|x-c\right|<\delta$,
then $F\left(x\right)>d$. But then $c-\delta$ is also an upper bound
for $Q_{d}$, which is again a contradiction. Hence, $F\left(c\right)=d$.
\end{proof}

\begin{thm}
[Extreme Value Theorem]\label{thm:-extreme} Let $F$ be a continuous
real $\i$- function on a non-degenerate interval $\left[a,b\right]$.
It is provable in $RST_{HF}^m$ that $F$ attains its maximum
and minimum.
\end{thm}

\begin{proof}
Let $Q$ be the $\i$-set $\left[a,b\right]\cap\mathbb{Q}$. $F\left[Q\right]$
is a $\i$-set by Prop. \ref{prop:replacement}, and it is non-empty
by the denseness of $\mathbb{Q}$ in $\mathbb{R}$. Assume by contradiction
that $F\left[Q\right]$ is not bounded, and define for every $n\in\mathbb{N}$
$C_{n}=\left\Vert \setin{x\in\tilde{Q}}{F\left(x\right)>n}\right\Vert $.
By the assumption $C_{n}$ is a non-empty, bounded $\i$-set. Therefore,
by Thm. \ref{prop:R-completeness}, each $C_{n}$ has a least upper
bound, denote it $c_{n}$. It is easy to see that $c_{n}\in\left[a,b\right]$
for each $n\in\mathbb{N}$. Now, define the $\i$-sequence $\lambda n\in\mathbb{N}.c_{n}$
(which is indeed a $\i$-sequence by Thm. \ref{prop:R-completeness}).
Standard arguments show that since $\left[a,b\right]$ is closed and
bounded, there is a subsequence of $\lambda n\in\mathbb{N}.c_{n}$,
$\lambda k\in\mathbb{N}.c_{n_{k}}$, which converges to a limit, denote
it $m$. By Prop. \ref{prop:closeCauchy-1-1} we have that $m\in\left[a,b\right]$.
Now, since $F$ is continuous, we easily get that $\lambda k\in\mathbb{N}.F\left(c_{n_{k}}\right)$
converges to $F\left(m\right)$. But, for each $k\in\mathbb{N}$:
$F\left(c_{n_{k}}\right)>n_{k}\geq k$, which contradicts the convergence
of the sequence. Hence, $F\left[Q\right]$ is bounded, and by Thm.
\ref{prop:R-completeness} it has a least upper bound, denote it by
$d$. Assume by contradiction that there exists $u\in\left[a,b\right]$
s.t. $F\left(u\right)>d$. Picking $\varepsilon=F\left(u\right)-d$,
the continuity of $F$ entails that there exists $\delta$ s.t. for
every $x\in B_{\delta}\left(u\right)$, $F\left(x\right)\geq d$.
But the denseness of $\mathbb{Q}$ entails that there is a rational
number $q\in B_{\delta}\left(u\right)$, and thus $F\left(q\right)\geq d$,
which is a contradiction. The proof that there exists $x\in\left[a,b\right]$
s.t. that $F\left(x\right)=d$ uses arguments similar to the ones
used in the proof of Thm. \ref{thm:Intermediate}. The proof that
$F$ attains its minimum is symmetric. 
\end{proof}

The next step is to introduce in $RST_{HF}^m$ the concepts of
differentiation, integration, power series, etc, and develop their
theories. It should now be clear that there is no difficulty in doing
so. Since a thorough exposition obviously could not fit in one paper 
we omit it here, but use some relevant facts 
 in what follows.

\bs

We end this section by showing that all elementary functions that
are relevant to $J_{2}$ are available in $RST_{HF}^{m}$ in the
sense that they are formalizable as $\i$-functions and their basic
properties are provable in $RST_{HF}^{m}$. Of course, not all
constant functions on the ``real'' real line are available in $J_{2}$,
even though for every $y$ in $\mathbb{R}$, $\lambda x\in\mathbb{R}.y$
is available in $RST_{HF}^{m}$ as a $\i$-function. The reason
is that $\lambda x\in\mathbb{R}.y$ does not exists in $J_{2}$ for
\emph{every} ``real'' number $y$ (for the simple fact that not
every ``real'' real number is available in $RST_{HF}^{m}$).
Thus we next define what is an ``$J_{2}$-elementary function''
(see, for example, \cite{Risch1979Analysis} for a standard definition
of ``elementary function'').
\begin{defi}
The collection of $J_{2}$-elementary functions is the minimal collection
that is closed under addition, subtraction, multiplication, division,
and composition, and includes the following:

\begin{itemize}
\item $J_{2}$-constant functions: $\lambda x\in\mathbb{R}.c$ where $c$
is a real number in $J_{2}$.
\item Exponential: $\lambda x\in\mathbb{R}.e^{x}$
\item Natural logarithm: $\lambda x\in\mathbb{R}^{+}.\mbox{ln}x$
\item Trigonometric functions:$\lambda x\in\mathbb{R}.\mbox{sin}x$.
\item Inverse trigonometric functions: $\lambda x\in\left[-1,1\right].\mbox{arcsin}x$.
\end{itemize}
\end{defi}
\begin{prop}
\label{prop:All--polynomials}All $J_{2}$-polynomials (i.e. with
coefficients in $J_{2}$) on $\mathbb{R}$ are available in $RST_{HF}^{m}$
as $\i$-function, and it is provable in $RST_{HF}^{m}$ that they
are continuous.
\end{prop}
\begin{proof}
$J_{2}$-constant functions and the identity function are available
in $RST_{HF}^{m}$ by Prop. \ref{prop:functions}, and the proofs
of their continuity is immediate. Composition of $\i$-functions is
also available in $RST_{HF}^{m}$. All $J_{2}$-polynomials on
$\mathbb{R}$ are therefore available in $RST_{HF}^{m}$, since
$+$ and $\cdot$ are $\i$-functions, and they are continuous by
Lemma \ref{prop:comp sum prod}.
\end{proof}
\begin{prop}
\label{prop:The-exponential-trigo}The exponential and trigonometric
functions are available in $RST_{HF}^{m}$, and it is provable
in $RST_{HF}^{m}$ that they are continuous.
\end{prop}
\begin{proof}
Since the exponential and the trigonometric functions all have power
series, their definability as $\i$-functions follows from Prop.
\ref{prop:Pointwise}. It is straightforward to verify that the basic
properties of these $\i$-functions are provable in $RST_{HF}^{m}$.
Examples of such properties are: the monotonicity of the exponential,
the power rules of the exponential, trigonometric identities like
$\mathrm{sin}\left(\alpha+\beta\right)=\mathrm{sin}\,\alpha\,\mathrm{cos\,}\beta\,+\mathrm{sin}\,\beta\mathrm{\,cos}\,\alpha$,
the fact that $\mathrm{sin}$ has a period of $2\pi$ (where $\pi$
is its first positive root), etc.\footnote{We can prove the standard properties of the exponent and the trigonometric
functions as listed, e.g., in \cite{ahlfors1964complex}, using the
notion of differentiation.} The continuity of these functions follows from Lemma \ref{prop:comp sum prod}
and Prop. \ref{prop:Pointwise}.
\end{proof}
\begin{lem}
Let $F$ be a continuous, monotone real $\i$-function on a real interval
$\left[a,b\right]$, and suppose $F\left(a\right)<F\left(b\right)$.
It is provable in $RST_{HF}^{m}$ that
\[
\forall y\in\bar{\mathbb{R}}\left(\exists x\in\overline{\left[a,b\right]}.\bar{F}\left(x\right)=y\leftrightarrow y\in\overline{\left[F\left(a\right),F\left(b\right)\right]}\right)
\]
\end{lem}
\begin{proof}
The left-to-right implication is immediate from the monotonicity of
$F$. The right-to-left implication follows from Thm. \ref{thm:Intermediate}.
\end{proof}
\begin{prop}
\label{prop:F-1}Let $F$ be a continuous, strictly monotone real
$\i$-function on a real interval. Then it is provable in $RST_{HF}^{m}$
that the inverse function $F^{-1}$ is available in $RST_{HF}^{m}$
as a $\i$-function, and its continuity is provable in $RST_{HF}^{m}$.
\end{prop}
\begin{proof}
We here prove the claim for continuous, strictly monotone real $\i$-function
on a finite closed interval $\left[a,b\right]$. The extension from
finite closed intervals to arbitrary interval is standard. Suppose
$F$ is increasing. The proof is similar to the proof of Thm. \ref{thm:Intermediate}.
For any $y\in\left[F\left(a\right),F\left(b\right)\right]$ define
the $\i$-set $Q_{y}:=\left\Vert \setin{q\in\widetilde{\mathbb{Q}}}{q\in\overline{\left[a,b\right]}\wedge F\left(q\right)\leq y}\right\Vert $.
It is easy to see that $Q_{y}$ is non-empty and bounded, thus, by
Thm. \ref{prop:R-completeness}, $Q_{y}$ has a least upper bound.
Now, $\left\Vert \lambda y\in\overline{\left[F\left(a\right),F\left(b\right)\right]}.l.u.b
(\widetilde{Q_{y}})\right\Vert $
is the desired inverse $\safe$-function. It is not difficult to prove
the basic properties of the inverse function in $RST_{HF}^{m}$.
We demonstrate the proof that $F^{-1}\circ F=id_{\left[a,b\right]}$.
For this we need to show that for any $x\in\left[a,b\right]$, $l.u.b(Q_{F\left(x\right)})=x$.
By the monotonicity of $F$, $x$ is clearly an upper bound for $Q_{F\left(x\right)}$.
Assume by contradiction that there is a real number $w<x$ which is
an upper bound of $Q_{F\left(x\right)}$. Thus, in the interval $\left(w,x\right)$
there is a rational number $q$ such that $F\left(q\right)\leq F\left(x\right)$
(by monotonicity). But then, $q\in Q_{F\left(x\right)}$ and $w<q$,
which is a contradiction.
\end{proof}
\begin{prop}
\label{prop:elementary-functions}All $J_{2}$-elementary functions
are available in $RST_{HF}^{m}$.
\end{prop}
\begin{proof}
Props. \ref{prop:All--polynomials} and \ref{prop:The-exponential-trigo}
show that $J_{2}$-polynomials on $\mathbb{R}$, the exponential and
trigonometric functions are available in $RST_{HF}^{m}$. Prop.
\ref{prop:F-1} then enables the availability in $RST_{HF}^{m}$
of the inverse trigonometric functions, and of the natural logarithm
as the inverse of the exponential.
\end{proof}
It is not difficult to see that many standard discontinuous functions
are also available in $RST_{HF}^{m}$, as the next proposition
shows.
\begin{prop}
Any piecewise defined function with finitely many pieces such that
its restriction to any of the pieces is a $J_{2}$-elementary function,
is available in $RST_{HF}^{m}$.
\end{prop}
\begin{proof}
If the function has finitely many pieces and each of the pieces is
a $J_{2}$-elementary function, then it can be constructed in $RST_{HF}^{m}$
using Prop. \ref{prop:functions}(5).
\end{proof}

\section{\label{sec:Further-Research}Conclusion and Further Research}

In this paper we showed that a minimal computational framework
is sufficient for the development of applicable mathematics. 
Of course, a major future research
task is to implement and test the framework. A critical component
of such implementation will be to scale the cost of checking the safety
relation. We then plan to use the implemented framework to formalize
even larger portions of mathematics, including first of all more analysis,
but also topology and algebra.

\ms

Another important task is to fully exploit the computational power of 
$RST_{HF}^{m}$ and $J_2$.
This includes finding a good notion of canonical terms, and investigating
various reduction properties such as strong normalization.
We intend to try also to profit from this computational power in other
ways, e.g., by using it for proofs by reflection as supported by well-known
proof assistant like Coq \cite{Chl13a}, Nuprl \cite{ConstableImplementing86}
and Isabelle/HOL \cite{nipkow2002isabelle}.

\ms

An intuitionistic
variant of the system $RST_{HF}^m$ can  
be also considered. It is based on intuitionistic first-order
logic (which underlies constructive counterparts of $ZF$,
like $CZF$ \cite{aczel2001CST} and $IZF$ \cite{beeson2012foundations}),
 and is obtained by adding to $RST_{HF}^m$ the axiom of Restricted
Excluded Middle: $\varphi\vee\neg\varphi$, where $\varphi\safe\emptyset$.
This axiom is computationally acceptable since it simply asserts the
definiteness of absolute formulas. The resulting computational theory  should 
allow for a similar formalization of constructive analysis (e.g.,
\cite{myhill1975CST}). 

\ms

Further exploration of the connection between our framework and other
related works is also required. This includes works on: computational
set theory \cite{aczel2001CST,beeson2012foundations,cantone2001set,Friedman77,myhill1975CST},
 operational set theory \cite{FEFERMAN2009971,jagerexpilcit2014},
and rudimentary set theory \cite{beckmann2015safe,mathias2015rudimentary}.

\ms

Another direction for further research is to consider larger computational
structures. This  includes $J_{\omega}$ or even $J_{\omega^{\omega}}$
(which is the minimal model of the minimal computational theory based
on ancestral logic \cite{avron2010new,cohen2015middle}). On the one
hand, in such universes standard mathematical structures can be treated
as sets. On the other hand, they are more comprehensive and less concrete,
thus include more objects which may make computations harder.

\section*{Acknowledgements}
The second author is supported by:
Fulbright Post-doctoral Scholar program; Weizmann Institute of Science
-- National Postdoctoral Award program for Advancing Women in Science;
Eric and Wendy Schmidt Postdoctoral Award program for Women in
Mathematical and Computing Sciences.
\bibliographystyle{plain}
\bibliography{j2}

\end{document}